\documentclass{fundam}
\publyear{22}
\papernumber{2102}
\volume{185}
\issue{1}

\usepackage{graphicx}
\usepackage{orcidlink}
\usepackage{float}
\usepackage{listings}
\lstdefinelanguage{pseudo}{morekeywords={init,with,or,if,then,else,fi,and,not,while,do,od,distinct,
    case, goto,local,algorithm, function, for, each, times, from, to,
    variables, procedure, recursive, return},
  morecomment=[l]{//}, morecomment=[s]{/*}{*/},
  mathescape=true,escapechar={@},
  basicstyle=\sffamily\small,
  commentstyle=\itshape\rmfamily\small,
  keywordstyle=\sffamily\bfseries\small
}
\usepackage{mathtools}
\usepackage{amsfonts,amssymb,amsmath}
\usepackage{thm-restate}

\usepackage{stmaryrd}
\usepackage{bm}

\usepackage{hyperref}
\usepackage{environ}
\usepackage{csvsimple}
\usepackage{longtable} 
\usepackage{fourier-orns} 
\usepackage{wasysym} 
\usepackage{marvosym} 
\usepackage{bbding} 

\usepackage{enumerate}
\usepackage{mathrsfs}
\usepackage{mathtools}
\usepackage{appendix}
\usepackage{comment}
\usepackage{tikz}
\usepackage{array}
\usetikzlibrary{calc}
\usepackage{multicol}

\usetikzlibrary{arrows,shapes,snakes,automata,backgrounds,petri,positioning}
\usetikzlibrary{fit,backgrounds} 
\definecolor{processblue}{cmyk}{0.96,0,0,0}

\makeatletter
\newcommand{\dashact}[2][]{\ext@arrow 0359\rightarrowfill@@{#1}{#2}}
\def\rightarrowfill@@{\arrowfill@@\relax\relbar\dashrightarrow}
\def\arrowfill@@#1#2#3#4{%
	$\m@th\thickmuskip0mu\medmuskip\thickmuskip\thinmuskip\thickmuskip
	\relax#4#1
	\xleaders\hbox{$#4#2$}\hfill
	#3$%
}
\makeatother

\newcommand{\nn}{\mathbb{N}}

\newcommand{\zn}{\mathbb{Z}}

\newcommand{\be}{\begin{enumerate}}
\newcommand{\ee}{\end{enumerate}}
\newcommand{\bc}{\begin{center}}
\newcommand{\ec}{\end{center}}

\newcommand{\bi}{\begin{itemize}}
\newcommand{\ei}{\end{itemize}}

\newcommand{\act}{\xrightarrow}

\newcommand{\node}{\mathsf{n}}

\usepackage[textsize=scriptsize]{todonotes}
\usepackage{hyperref}

\newcommand{\NP}{\textsf{NP}}
\newcommand{\coNP}{\textsf{coNP}}

\newcommand{\bx}{\mathbf{x}}
\newcommand{\by}{\mathbf{y}}

\newcommand{\TA}{\textsf{TA}}
\newcommand{\Env}{\text{Env}}
\newcommand{\local}{{\mathcal L}}
\newcommand{\initlocal}{{\mathcal I}}
\newcommand{\paraset}{\Pi}
\newcommand{\globset}{\Gamma}

\newcommand{\ruleset}{{\mathcal R}}
\newcommand{\ResCond}{{\mathit{RC}}} 
\newcommand{\syssize}{N}

\newcommand{\update}{\vec{u}}

\newcommand{\fromstate}{{\mathit{from}}}
\newcommand{\tostate}{{\mathit{to}}}

\newcommand{\precond}{\varphi} 

\newcommand{\counters}{{\vec{\boldsymbol\kappa}}}
\newcommand{\vars}{\vec{g}}
\newcommand{\param}{\mathbf{p}}

\newcommand{\statectx}{\omega}

\newcommand\slice[2]{#1{\raise-.5ex\hbox{\ensuremath|}}_{#2}}

\newcommand\ELTLFT{\textsf{ELTL}_\textsf{FT}}

\newcommand{\cpp}[1]{#1{\footnotesize{\texttt{++}}}}

\usepackage{url} 
\usepackage[ruled,lined]{algorithm2e}
\usepackage{graphicx}
\usepackage{tikz}
\usetikzlibrary{automata, positioning, arrows}

\begin{document}
	\title{Coefficient Synthesis for Threshold Automata}
	
	\address{Max Planck Institute for Software Systems, Kaiserslautern, Germany.}
	
	\author{A. R. Balasubramanian
		\thanks{A part of the work was done when the author was at Technische Universit\"{a}t M\"{u}nchen, Munich, Germany}\\
		Max Planck Institute for Software Systems \\ 
		Kaiserslautern, Germany\\
		bayikudi@mpi-sws.org}
\maketitle              
\runninghead{A. R. Balasubramanian}{Coefficient Synthesis for Threshold Automata}
\begin{abstract}
	Threshold automata are a formalism for modeling fault-tolerant distributed algorithms.
	The main feature of threshold automata is the notion of a threshold guard, which
	allows us to compare the number of received messages with the 
	total number of different types of processes.
	In this paper, we consider the coefficient synthesis problem for threshold automata, in which 
	we are given a sketch of a threshold automaton (with some of the constants in the threshold guards left unspecified) and a violation describing a collection of undesirable behaviors. We then want to synthesize a set of constants which
	when plugged into the sketch, gives a threshold automaton that does not have the undesirable behaviors. 
	Our main result is that this problem is undecidable, even when the violation
	is given by a coverability property and the underlying sketch is acyclic. 
	
	We then consider the
	bounded coefficient synthesis problem, in which a bound on the constants to be synthesized is also provided.
	Though this problem is known to be in the second level of the polynomial hierarchy for coverability properties, the algorithm for this problem 
	involves an exponential-sized encoding of the reachability relation
	into existential Presburger arithmetic. In this paper, we give a polynomial-sized encoding for this relation. 
	We also provide a tight complexity lower bound for this problem against coverability properties. Finally, motivated by benchmarks appearing from the literature, we also consider a special class of threshold automata
	and prove that the complexity decreases in this case. 
\end{abstract}


\section{Introduction}
Threshold automata~\cite{KVW17:IandC} are a formalism for modeling and analyzing parameterized fault-tolerant distributed algorithms.
In this setup, an arbitrary but finite number of processes execute a given distributed protocol modeled as a threshold automaton. 
Verifying these systems amounts to proving that the given protocol is correct with respect to a given specification, irrespective
of the number of agents executing the protocol. Many algorithms have been developed for verifiying properties of threshold automata~\cite{AllFlavors,BKLW19,ELTLFT,KW18,KVW17:IandC,FSTTCS20} and it is known that the reachability problem for threshold automata is \NP-complete~\cite[Theorem 1 and Corollary 1]{Complexity}.

In many formalisms for modeling distributed systems (like rendez-vous protocols~\cite{GermanS92} and reconfigurable broadcast networks~\cite{AdHoc}), the status of a transition being enabled or not
depends only on a fixed number of processes, independent of the total number of participating processes. 
One of the central features that distinguishes threshold automata from such formalisms is the notion of a \emph{threshold guard}. 
A threshold guard can be used to specify relationships between the number of messages received and the total number of participating
processes, in order for a transition to be enabled. For example, if we let $x$ be a variable counting the number of messages of a
specified type, $n$ be the number of participating processes and $t$ be the maximum number of processes which can fail, 
then the guard $x \ge n/3 + t$ on a transition specifies that the number of messages received should be at least $n/3+t$, in order for 
a process to execute this transition. 

While the role of these guards is significant for the correctness of these protocols, they can also be unstable as small changes (and hence small calculation errors)
in the coefficients of these guards can make a correct protocol faulty. 
(A concrete example of this phenomenon will be illustrated in the next section). 
For this reason, it would be desirable to automate
the search for coefficients so that once the user gives a ``sketch'' of a threshold automaton (which only specifies the control flow but leaves out some of the arithmetic details)
and a violation property (such as the property of being able to put a process in some error state), we can compute a set of coefficient values, which when ``plugged into'' the sketch does not have any of the behaviors given by the violation property.
With this motivation, the authors of~\cite{LKWB17:opodis} 
tackle this \emph{coefficient synthesis problem} and provide theoretical and experimental results. They show
that for a class of ``sane'' threshold automata, this problem is decidable against a particular class of properties and provide a CEGIS approach for synthesizing these coefficients. However, the decidability status of the coefficient synthesis problem for the general case has remained open so far.

In this paper, we prove that this problem is actually undecidable, hence settling the decidability status of this 
problem. We do this by giving a reduction from 
a fragment of \emph{Presburger arithmetic with divisibility}, for which the
validity problem is known to be undecidable. Further, our result already shows that the coefficient synthesis
problem is undecidable, even when the violation property simply specifies that a given state can be populated by some process (a so-called coverability property)
and 
the underlying control-flow structure of the sketch automaton is acyclic.


We then consider the \emph{bounded coefficient synthesis} problem~\cite{Complexity}, where in addition to the sketch and the
property, the user also gives a bound (in the form of an interval) on the coefficients.
For violations specified by coverability properties, this problem was already known to be in $\Sigma_2^p$~\cite[Theorem 7]{Complexity}.
The main ingredient that was used to prove this upper bound was that the reachability relation
of a threshold automaton can be defined using an existential Presbuger formula.
However, the size of the formula given in~\cite{Complexity} can be exponential in the size of the input.
Our second main result is that we can efficiently construct a formula in existential Presburger arithmetic of polynomial size for the same task. Furthermore, we also provide a matching lower bound for bounded coefficient synthesis against coverability properties.
Finally, motivated by benchmarks appearing in the literature, we consider a special class of threshold automata and prove that bounded coefficient synthesis for this class against 
coverability properties is \coNP-complete.

\subsubsection*{Related work.} As mentioned before, the coefficient synthesis problem has already been studied in~\cite{LKWB17:opodis}. However, the decidability status of the general case was left open in that paper and
here we show it is undecidable. 
A similar problem has also been studied for \emph{parametric timed automata}~\cite{AlurHV93}, where the
control flow of a timed automaton is given as input and 
we have to synthesize coefficients for the guards in order to satisfy a given reachability specification.
The authors show that the problem is undecidable, already for timed automata with three clocks. They also show that it is
decidable when the automaton has only one clock. 
Unlike clocks, the shared variables in our setting cannot be reset. 
Further, in our setting, variables can be compared with both the coefficients and other environment variables, which is 
not the case with parametric timed automata.

\subsubsection*{Organization of the paper. } The rest of the paper is organized as follows. In Section~\ref{sec:prelims}, we recall the necessary definitions
needed to state the coefficient synthesis problem. In Section~\ref{sec:undec}, we define the fragment of Presburger arithmetic
with divisibility that we will be working with and we prove our main undecidability result by giving a reduction from this fragment. Then, in Section~\ref{sec:bound-synth}, we present all our results regarding the bounded coefficient synthesis problem and the reachability relation of a threshold automaton. We conclude in Section~\ref{sec:concl}.


\section{Preliminaries}\label{sec:prelims}
Let $\nn_{>0}$ be the set of positive integers and 
$\nn$ be the set of non-negative integers. 

\subsection{Threshold Automata}

We introduce threshold automata, mostly following the definitions and notations
used in~\cite{Complexity,FSTTCS20}. Along the way, we also 
illustrate the definitions on the example of Figure~\ref{fig:stexample} from~\cite{ELTLFT}, which is a model of the Byzantine agreement protocol of Figure \ref{fig:st}. 

\begin{figure}[t]
	\begin{minipage}{.53\textwidth}
		\input{st87-pseudo.tex}
		\caption{Pseudocode of a reliable broadcast protocol from~\cite{ST87:abc}
			for a correct process~$i$, where $n$ and $t$ denote the number of processes,
			and an upper bound on the number of faulty processes. 
			The protocol satisfies its specification 
			(if $\mathit{myval}_i=0$ for every correct process $i$, then
			no correct process sets its $\mathit{accept}$ variable to $\mathit{true}$) if $t < n/3$. }\label{fig:st} 
	\end{minipage}
	\begin{minipage}{.02\textwidth}
		\phantom{x}
	\end{minipage}
	\begin{minipage}{.45\textwidth}
		\centering 
		\tikzstyle{node}=[circle,draw=black,thick,minimum size=4.3mm,inner sep=0.75mm,font=\normalsize]
\tikzstyle{init}=[circle,draw=black!90,fill=green!10,thick,minimum size=4.3mm,inner sep=0.75mm,font=\normalsize]
\tikzstyle{final}=[circle,draw=black!90,fill=red!10,thick,minimum size=4.3mm,inner sep=0.75mm,font=\normalsize]
\tikzstyle{rule}=[->,thick]
\tikzstyle{post}=[->,thick,rounded corners,font=\normalsize]
\tikzstyle{pre}=[<-,thick]
\tikzstyle{cond}=[rounded
  corners,rectangle,minimum
  width=1cm,draw=black,fill=white,font=\normalsize]
\tikzstyle{asign}=[rectangle,minimum
  width=1cm,draw=black,fill=gray!5,font=\normalsize]

\tikzset{every loop/.style={min distance=5mm,in=140,out=113,looseness=2}}
\begin{tikzpicture}[>=latex, thick,scale=1.1, every node/.style={scale=1}]

\node[] at (0, 0.85) [init,label={[label distance=-0.5mm]180:\textcolor{blue}{$\ell_0$}}]         (0) {};
 \node[] at (0, -0.85) [init,node,label={[label distance=-0.5mm]180:\textcolor{blue}{$\ell_1$}}]   (1) {};

 \node[] at (2.75, 0) [node,label=below:\textcolor{blue}{$\ell_2$}]        (2) {};
 \node[] at (4.2, 0) [final,label=below:\textcolor{blue}{$\ell_3$}]    (3) {};

\draw[post] (0) to[]
    node[sloped, above, align= center,xshift=0cm]
    {\small $r_2 \colon \gamma_1 \mapsto \cpp{x}\quad$} (2);
\draw[post] (1) to[] node[sloped, above, align= center,yshift=0cm]
    {\small $ r_1 \colon true \mapsto \cpp{x}$~ ~}(2);
\draw[post] (2)to[]
    node[align=center,anchor=north, midway]
    {\small $r_3 \colon \gamma_2$} (3);

\end{tikzpicture}\caption{Threshold automaton from~\cite{ELTLFT} modeling the body of the loop in the 
			protocol from Fig.~\ref{fig:st}. Symbols $\gamma_1, \gamma_2$ stand for the
			threshold guards $x \ge (t+1) - f$ and $x \ge (n-t) - f$, 
			where $n$ and $t$ are as in Fig.~\ref{fig:st}, and $f$ is the actual
			number of faulty processes. The shared variable $x$ models the number of ECHO messages sent by 
			correct processes. Processes with $\mathit{myval}_i=b$ (line~1) start in location 
			$\ell_b$ (in green). Rules $r_1$ and $r_2$ model sending ECHO at lines 7
			and 12. 
		}\label{fig:stexample}
	\end{minipage}
\end{figure}

\subsubsection*{Environment. } An \emph{environment} is a tuple $\Env = (\paraset, \ResCond, \syssize)$, where $\paraset$ is a finite set of \emph{environment variables} ranging over $\nn$, $\ResCond \subseteq \nn^{\paraset}$ is a \emph{resilience condition}
over the environment variables, given as a linear formula, and $\syssize \colon \ResCond \rightarrow \nn$ is a linear function
called the \emph{number function}. Intuitively, an assignment of $\paraset$ determines the number of processes of different kinds 
(e.g.\ faulty) executing the protocol, and $\ResCond$ describes the admissible combinations of values of environment variables. 
Finally, $\syssize$ associates to a each admissible combination, the number of processes explicitly modeled. 
In a Byzantine setting, faulty processes behave arbitrarily, and so we do not model them explicitly;
In the crash fault model, processes behave correctly until they crash and they must be modeled explicitly.

\begin{example}
In the threshold automaton of Figure~\ref{fig:stexample}, the environment variables are $n$, $f$, and $t$, describing the number of processes, the number of (Byzantine) faulty processes, and the maximum possible number of faulty processes, respectively. The resilience condition is the constraint $n/3 > t \geq f$.
The function $\syssize$ is given by $\syssize(n,t,f)= n-f$, which is the number of correct processes. 
\end{example}

\subsubsection*{Threshold automata. } 
A \emph{threshold automaton} over an environment $\Env$ is a tuple $\TA=(\local, \initlocal, \globset, \ruleset)$, where
$\local$ is a finite set of  \emph{local states} (or \emph{locations}), $\initlocal\subseteq\local$ is a nonempty subset  of \emph{initial locations}, $\globset$ is a finite set of \emph{shared variables} ranging over $\nn$, and $\ruleset$ is a finite set of \emph{transition rules} (or just \emph{rules}), formally described below.

A \emph{transition rule} (or just a \emph{rule}) is a tuple $r = (\fromstate, \tostate,  \precond, \update)$, where $\fromstate, \tostate \in \local$ are the \emph{source} and \emph{target} locations respectively, $\precond
\subseteq \nn^{\paraset \cup \globset}$ 
is a conjunction of \emph{threshold guards} (described below), and $\update \colon \globset \rightarrow \{0,1\}$ is an \emph{update}. We often let $r.\fromstate, r.\tostate, r.\precond, r.\update$ denote the components of $r$.
Intuitively, $r$ states that a process can move from $\fromstate$ to $\tostate$ if the current values of $\paraset$ and $\globset$ satisfy $\varphi$, and when it moves, it updates the current valuation $\vars$ of $\globset$ by performing the update $\vars := \vars + \update$. Since all components of $\update$ are nonnegative, the values of shared variables never decrease.  A \emph{threshold guard} $\varphi$ has one of the following forms: $b \cdot x \bowtie a_0 + a_1 \cdot p_1 + \ldots + a_{|\paraset|} \cdot p_{|\paraset|}$
where~$\bowtie \ \in \{\ge,>,=,<,\le\}$, $x$ $\in \globset$ is a shared variable, 
$p_1,\ldots, p_{|\paraset|}\in \paraset$ are the environment variables,
$b \in \nn_{> 0}$ and $a_0,a_1,\ldots,a_{|\paraset|}\in \zn$ are integer coefficients.

The underlying graph of a threshold automaton is the graph obtained by taking the vertices as the locations 
and connecting any two vertices with an edge as long as there is a rule between them. A threshold automaton is 
called \emph{acyclic} if its underlying graph is acyclic. 
\begin{example}
The threshold automaton from Figure~\ref{fig:stexample} is acyclic.
The rule $r_3$ of this automaton has $\ell_2$ and $\ell_3$ as its source and target
locations, $x \geq (n-t)-f$ as its guard, and does not increment any shared variable. On the other hand,
the rule $r_1$ has $\ell_1$ and $\ell_2$ as its source and target locations, no guard (denoted by $true$) and increments the 
variable $x$. 
\end{example}

\subsubsection*{Configurations and transition relation. }
A \emph{configuration} of $\TA$ is a triple $\sigma=(\counters,\vars,\param)$ where 
$\counters  \colon \local \rightarrow \nn$ describes the number of processes at each location, and $\vars \in \nn^{\globset}$ and $\param  \in \ResCond$ are valuations of the shared variables and the environment variables respectively. In particular, $\sum_{\ell \in \local} \counters(\ell)= \syssize(\param)$ always holds.  A configuration is \emph{initial} if $\counters(\ell) =0$ for every $\ell \notin \initlocal$, and $\vars = \vec{0}$. We often let $\sigma.\counters, \sigma.\vars, \sigma.\param$ denote the components of $\sigma$. 

A configuration $\sigma=(\counters,\vars,\param)$  
\emph{enables} a rule $r = (\fromstate, \tostate,  \precond, \update)$ if $\counters(\fromstate) > 0$, and $(\vars, \param)$ satisfies the guard $\precond$, i.e., substituting $\vars(x)$ for $x$ and $\param(p_i)$ for $p_i$ in $\precond$ yields a true expression, denoted by $\sigma\models\varphi$. If $\sigma$ enables $r$, then there is a \emph{step} from $\sigma$ to the configuration $\sigma'=(\counters',\vars',\param')$ given by,
(i)  $\param' = \param$, (ii) $\vars' = \vars + \update$, and (iii) $\counters' = \counters + \vec{v}_r$, where $\vec{v}_r = \vec{0}$ if
$\fromstate = \tostate$ and otherwise,
$\vec{v}_r (\fromstate) =-1$, $\vec{v}_r (\tostate) =+1$, and $\vec{v}_r(\ell) = 0$ for all other locations $\ell$.
We let $\sigma \act{r} \sigma'$ denote that $\TA$ there is a step from $\sigma$ to $\sigma'$ using the rule $r$.
We use $\sigma \act{} \sigma'$ to denote that $\sigma \act{r} \sigma'$ for some rule $r$. 

A schedule is a finite sequence of rules. Given a schedule $\tau = r_1,r_2,\dots,r_k$ and two configurations $\sigma, \sigma'$ we say that $\sigma \act{\tau} \sigma'$ if 
there exist configurations $\sigma_0, \sigma_1, \dots, \sigma_k$ such that $\sigma_0 = \sigma$,
$\sigma_{i-1} \act{r_i} \sigma_{i}$ for all $0 < i \le k$ and $\sigma_k = \sigma'$. 
In this case, we will call the sequence $\sigma_0, \sigma_1, \dots, \sigma_k$ a \emph{run} or a \emph{path} between $\sigma$
and $\sigma'$.

We let $\sigma \act{*} \sigma'$ to mean that $\sigma \act{\tau} \sigma'$ for some schedule $\tau$.
If $\sigma \act{*} \sigma'$, we say 
that $\sigma'$ is reachable from $\sigma$.

\subsubsection*{Coverability. }\label{sec:spec}

Let $\ell \in \local$ be a location. We say that a configuration $\sigma$ covers 
$\ell$ if $\sigma(\ell) > 0$. We say that $\sigma$ can cover a location $\ell$ if $\sigma$
can reach a configuration $\sigma'$ such that $\sigma'$ covers $\ell$. Finally, we say that $\TA$ can cover
$\ell$ if some initial configuration of $\TA$ can cover $\ell$. Hence, $\TA$ \emph{cannot cover} $\ell$ 
if and only if every initial configuration of $\TA$ cannot cover $\ell$. It is known that deciding whether
a given threshold automaton can cover a given location is \NP-complete~\cite[Theorem 1 and Corollary 1]{Complexity}.

\subsubsection*{Coefficient synthesis. }\label{subsec:synthesis}

We now come to the definition of the main problem that we will be interested in this paper,
namely the coefficient synthesis problem~\cite{LKWB17:opodis}. 
To introduce this problem, we first have to introduce the notion of \emph{sketch threshold automata}, which we do now.


\paragraph*{Sketch threshold automata.} In a threshold automaton, a guard is an inequality which can be of the form $b \cdot x \bowtie a_0 + a_1 \cdot p_1 + \dots a_{|\paraset|} \cdot p_{|\paraset|}$ with $\bowtie \ \in \{\ge,>,=,<,\le\}, \ b \in \nn_{>0}$ and $a_1,\dots,a_{|\paraset|} \in \zn$. 
A \emph{sketch threshold automaton} (or simply a \emph{sketch}) is the same as a threshold automaton, except that some of the $b, a_0, a_1, \dots, a_{|\paraset|}$ terms in
any guard of the automaton are now allowed to be \emph{indeterminates}, 
which are variables that can take any integer values. Intuitively, a sketch threshold automaton completely specifies the
control flow of the protocol, but leaves out some of the precise arithmetic details of the threshold guards. 

Given a sketch $\TA$ and an integer assignment $\mu$ to the indeterminates, we let $\TA[\mu]$ denote the threshold automaton 
obtained by replacing the indeterminates with their corresponding values in $\mu$. The \emph{coefficient synthesis problem against coverability properties} for 
threshold automata is now defined as the following problem:
%
\begin{quote}
	\emph{Given: } An environment $\Env$, a sketch $\TA$ and a location $\ell$\\
	\emph{Decide: } Whether there is an assignment $\mu$ to the indeterminates such that $\TA[\mu]$ \textit{cannot} cover $\ell$.
\end{quote}
\begin{remark}\label{rem:cov}
	The coefficient synthesis problem as defined in~\cite{LKWB17:opodis} is a more general problem
	than the one that will be defined here. In that paper, along with an environment and a sketch, a violation is also given as input, where a violation is a collection of behaviors specified in a logic called $\ELTLFT$.
	The question then is to find an assignment to the indeterminates so that plugging in the assignment
	results in a threshold automaton that avoids the behaviors specified by the violation.
	Since coverability can be specified in that logic, it follows that our formulation is a special
	case of that formulation. 
	
	In this paper, for the sake of simplicity and presentation, we will only restrict ourselves to  coverability violations. Since the main 
	result of this paper is an undecidability result, this will also translate to the general case. In fact, all of our results, except the final one, all translate to the general case as well.
	A more detailed discussion regarding this point could also be found in the conclusion of this paper.
\end{remark}

\begin{example} 
	We consider the threshold automaton from Figure~\ref{fig:stexample}.
	As mentioned in the text under Figure~\ref{fig:st}, if no (correct) process initially starts at 
	$\ell_1$, then no process can ever reach $\ell_3$. This implies that if we remove the location $\ell_1$
	in the threshold automaton of Figure~\ref{fig:stexample}, then the modified threshold automaton $\TA'$
	will never be able to cover $\ell_3$.
	
	We can now convert $\TA'$ into a sketch, by replacing the
	guard $\gamma_1$ with $x \ge (t + a) - f$, where $a$ is an indeterminate. 
	When $a = 1$, we get $\TA'$ and so no reachable configuration has a process at the location $\ell_3$.
	However, when $a = 0$, this is not the case.
	Indeed if we set $n = 6$, $t = f = 1$ and if all the $N(6,1,1) = 6-1 = 5$ processes start at $\ell_0$ initially,
	then the guard $\gamma_1$ will always be true and so all the 5 processes can move to $\ell_2$, thereby setting the value
	of $x$ to 5. At this point, the guard $\gamma_2$ becomes true and so all the processes can move to $\ell_3$. This indicates that very small changes in the coefficients can make a protocol faulty.\\
\end{example}

Having stated all the necessary definitions, we now move on to the first result of this paper.

\section{Undecidability of Coefficient Synthesis}\label{sec:undec}
The first main result that we shall prove in this paper is that

\begin{theorem}~\label{th:main-th}
	The coefficient synthesis problem against coverability properties is undecidable, even for acyclic threshold automata.
\end{theorem}

\begin{remark}
	As mentioned in Remark~\ref{rem:cov}, in this paper we will only be concerned with violations
	specified by coverability properties. For this reason, in the sequel, we will refer to the coefficient
	synthesis problem against coverability properties as simply the coefficient synthesis problem.
\end{remark}


Theorem~\ref{th:main-th} is proved in two steps. First, we consider a restricted version
of the coefficient synthesis problem, called the \emph{non-negative coefficient synthesis problem}, in which given a tuple $(\Env,\TA,\ell)$, we want to find a \emph{non-negative assignment} $\mu$ to the
indeterminates so that the resulting automaton $\TA[\mu]$ does not cover $\ell$. 
We first show that the non-negative coefficient synthesis problem is undecidable. Then, we reduce non-negative coefficient synthesis to coefficient synthesis, thereby achieving the desired result.

To prove that the non-negative coefficient synthesis problem is undecidable, we give a reduction from the validity problem for a restricted fragment of \emph{Presburger arithmetic with divisibility}, which is known to be undecidable. We now proceed to formally define this fragment.

\subsection{Presburger Arithmetic with Divisibility}

We now recall the necessary definitions for introducing Presburger arithmetic with divisibility. We will mostly follow the notations given in~\cite{SynthesisOneCounter}.

\emph{Presburger arithmetic} (PA) is the first-order theory over $\langle \nn, 0, 1, +, < \rangle$ where $+$ and $<$ are the standard addition and order operations over the natural numbers $\nn$ with constants 0 and 1 interpreted in the usual way.
We can, in a straightforward manner, extend our syntax with the following abbreviations: $\le, = , \ge, >$ and $ax = \sum_{1 \le i \le a} x$ where $a \in \nn_{>0}$ and $x$ is a variable.
A \emph{linear polynomial} is an expression of the form $\sum_{1 \le i \le n} a_i x_i + b$ where each $x_i$ is a variable, each $a_i$ belongs to $\nn_{>0}$ and $b$ belongs to $\nn$. An \emph{atomic formula} is a formula of the form $p(\bx) \bowtie q(\bx)$
where $p$ and $q$ are linear polynomials over the variables $\bx$ and $\bowtie \ \in \{<,\le,=,>,\ge\}$. 

\emph{Presburger arithmetic with divisibility} (PAD) is the extension of PA obtained by adding a divisibility
predicate $|$ which is interpreted as the usual divisibility relation among numbers. For the purposes of this paper, we restrict ourselves to
the $\forall \exists_R \text{PAD}^{+}$ fragment of PAD, i.e., we shall only consider
statements of the form
\begin{equation}\label{eq:undec}
	\forall x_1,\dots,x_n \ \exists y_1,\dots,y_m \ \bigvee_{i \in I} \left(
	\bigwedge_{(j,k) \in S_i} (x_j | y_k) \land \bigwedge_{l \in B_i} A_l(x_1,\dots,x_n,y_1,\dots,y_m) \right)
\end{equation}
where $I, S_i, B_i$ are finite sets of indices and each $A_l$ is a quantifier-free atomic PA formula. 
It is known that checking if such a statement is true is undecidable~\cite{SynthesisOneCounter}.
We now prove the undecidability of the non-negative coefficient synthesis problem by a reduction from this problem.

\begin{remark}
	PAD as defined here allows us to quantify the variables only over the natural numbers, whereas in~\cite{SynthesisOneCounter} the 
	undecidability result is stated for the variant where the variables are allowed to take integer values. However, the same
	proof given in~\cite{SynthesisOneCounter} allows us to prove the undecidability result over the natural numbers as well.
\end{remark}

\begin{remark}
	In our definition of $\forall \exists_R \text{PAD}^+$, we only allow divisibility constraints of the form $x_j | y_k$. In~\cite{SynthesisOneCounter},
	divisiblity constraints of the form $f(\bx) | g(\bx,\by)$ were allowed, where $f$ and $g$ are any linear polynomials. 
	This does not pose a problem, because of the fact that $\forall x_1,\dots,x_n \ \exists y_1,\dots,y_m \  f(\bx) | g(\bx,\by)$ is true if and only if
	$\forall x_1,\dots,x_n,z \ \exists y_1,\dots,y_m,z' \ (z \neq f(\bx)) \lor (z = f(\bx) \ \land \ z' = g(\bx,\by) \ \land \ z | z')$. 
	Because of this identity, it is then clear that any formula in the $\forall \exists_R \text{PAD}^+$ fragment as defined in~\cite{SynthesisOneCounter}
	can be converted into a formula in our fragment without changing its validity.
\end{remark}

\subsection{The Reduction}

Let $\xi(x_1,\dots,x_n,y_1,\dots,y_m)$ be a formula of the form~\ref{eq:undec} with
$\bx$ denoting the collection $x_1,\dots,x_n$ and $\by$ denoting the collection $y_1,\dots,y_m$. 
The set of atomic formulas of $\xi$ is the set comprising
each quantifier-free atomic PA formula in $\xi$ and all the divisibility constraints of the form $x_j | y_k$ that appear in $\xi$.
The desired reduction now proceeds in two stages.

\paragraph*{First stage: The environment.} We begin by defining the environment  $\Env = (\paraset, \ResCond, \syssize)$. We will 
have $m$ environment variables $t_1,\dots,t_m$, with each $t_i$ intuitively corresponding to the variable $y_i$ in $\xi$. 
Further, for every atomic formula $A$ of $\xi$ which is a divisibility constraint, we will have an environment variable $d_A$.
Finally, we will have an environment variable $z$, which will intuitively denote the total number of participating processes.

The resilience condition $\ResCond$ will be the trivial condition $true$. The linear function $N : \ResCond \to \nn$
is taken to be $N(\Pi) = z$.  Hence, the total number of processes executing the threshold automaton will be $z$. 

\paragraph*{Second stage: The indeterminates and the sketch.} 
For each variable $x_i$ of $\xi$, we will have an indeterminate $s_i$. Before we proceed with the description of the sketch, we make a remark.

\begin{remark}
	Throughout the reduction, a \emph{simple} configuration of a sketch will mean
	a configuration $C$ such that 1) there is a \textbf{unique location} $\ell$ with $C(\ell) > 0$ and
	2) $C(v) = 0$ for every shared variable $v$, i.e., all the processes of $C$ are in exactly one location
	and the value of each shared variable is 0.
\end{remark}

We now proceed with the description of the sketch. 
Throughout the reduction, we let $\mathbf{s}$ denote the set of indeterminates $s_1,\dots,s_n$ and 
$\mathbf{t}$ denote the set of environment variables $t_1,\dots,t_m$ .
The sketch will now be constructed in three phases, which are as follows.

\subsubsection{First phase. } For each atomic formula $A$ of 
$\xi$, we will construct a sketch $\TA_A$.
$\TA_A$ will have a single shared variable $v_A$.
We now have two cases:

\begin{itemize}
	\item Suppose $A$ is of the form $x_j | y_k$ for some $j \in \{1,\dots,n\}$ and $k \in \{1,\dots,m\}$. Then, corresponding
	to $A$, we construct the sketch in Figure~\ref{fig:case-one}.
	\begin{figure}
		\begin{center}
			\tikzstyle{node}=[circle,draw=black,thick,minimum size=12mm,inner sep=0.75mm,font=\normalsize]
			\tikzstyle{edgelabelabove}=[sloped, above, align= center]
			\tikzstyle{edgelabelbelow}=[sloped, below, align= center]
			\begin{tikzpicture}[->,node distance = 2cm,scale=0.8, every node/.style={scale=0.8}]
				\node[node] (start) {$start_A$};
				\node[node, right = of start] (ell) {$\ell_A$};
				\node[node, right = of ell, xshift = 3cm] (end) {$end_A$};
				
				\draw(start) edge[edgelabelabove, bend left = 15] node{$\cpp{v_A}$} (ell);
				\draw(start) edge[edgelabelbelow, bend right = 15] node{} (ell);
				\draw(ell) edge[edgelabelabove] node{$v_A = s_j \cdot d_A \ \land \ v_A = t_k$} (end);
			\end{tikzpicture}
		\end{center}
		\caption{Sketch for the first case}
		\label{fig:case-one}
	\end{figure}
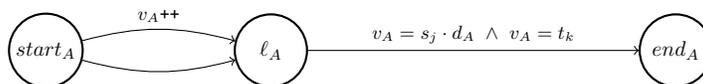 
	\item Suppose $A$ is of the form $f(\bx,\by) \bowtie g(\bx,\by)$ where $f$ and $g$ are linear polynomials and $\bowtie \ \in \{<,\le,=,>,\ge\}$.
	Then, corresponding to $A$, we construct the sketch in Figure~\ref{fig:case-two}.
	\begin{figure}
			\begin{center}
			\tikzstyle{node}=[circle,draw=black,thick,minimum size=12mm,inner sep=0.75mm,font=\normalsize]
			\tikzstyle{edgelabelabove}=[sloped, above, align= center]
			\tikzstyle{edgelabelbelow}=[sloped, below, align= center]
			\begin{tikzpicture}[->,node distance = 2cm,scale=0.8, every node/.style={scale=0.8}]
				\node[node] (start) {$start_A$};
				\node[node, right = of start] (ell) {$\ell_A$};
				\node[node, right = of ell, xshift = 3cm] (end) {$end_A$};
				
				\draw(start) edge[edgelabelabove, bend left = 15] node{$\cpp{v_A}$} (ell);
				\draw(start) edge[edgelabelbelow, bend right = 15] node{} (ell);
				\draw(ell) edge[edgelabelabove] node{$v_A = f(\mathbf{s},\mathbf{t}) \ \land \ v_A \bowtie g(\mathbf{s},\mathbf{t})$} (end);
			\end{tikzpicture}
		\end{center}
		\caption{Sketch for the second case}
		\label{fig:case-two}
	\end{figure}
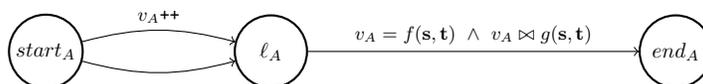	
\end{itemize}

\begin{remark}
	Notice that any assignment to the variables of $\mathbf{x}$ (resp. $\mathbf{y}$) can be interpreted in a straightforward manner
	as an assignment to $\mathbf{s}$ (resp. $\mathbf{t}$) and also vice versa. We will use this convention throughout the reduction.		
\end{remark}

Now let us give an intuitive idea behind the construction of these gadgets.
Intuitively, in both these cases, all the processes initially start at $start_A$. Then each process either
takes the top transition and increments $v_A$ or takes the bottom transition and does not increment any variable. Ultimately, this would lead to a point where all the processes are now at $\ell_A$. 
Then, in the first case, the guard from $\ell_A$ to $end_A$ essentially checks
that $s_j$ divides $t_k$ and in the second case, the guard from $\ell_A$ to $end_A$ essentially checks
that $f(\mathbf{s},\mathbf{t}) \bowtie g(\mathbf{s},\mathbf{t})$.
By the previous remark, the variables $\mathbf{s}$ (resp. $\mathbf{t}$) can be thought of as corresponding to the variables
$\mathbf{x}$ (resp. $\mathbf{y}$) and so this means that a process can reach $end_A$ if and only if $A$ can be
satisfied. We now proceed to formalize this intuition.

\begin{lemma}\label{prop:first-phase}
	Let $X$ and $Y$ be assignments to the variables $\mathbf{x}$ and $\mathbf{y}$ respectively. Then $A(X,Y)$ is true if and only if
	there is a 
	simple configuration $C$ of $\TA_A[X]$ with $C(start_A) > 0$ and $C(\mathbf{t}) = Y$ such that it can cover $end_A$.	
\end{lemma}

\begin{proof}
	($\Rightarrow$): Assume that $A(X,Y)$ is true. 
	\begin{itemize}
		\item Suppose $A$ is a divisibility constraint of the form $x_j | y_k$.
		Let $q$ be such that $X(x_j) \cdot q = Y(y_k)$ and let $C$ be the (unique) simple configuration
		given by $C(start_A) = C(z) = Y(y_k) + 1, C(t_k) = Y(y_k)$ and $C(d_A) = q$.
		\item Suppose $A$ is of the form $f(\bx,\by) \bowtie g(\bx,\by)$. 
		Let $C$ be the (unique) simple configuration
		given by $C(start_A) = C(z) = f(X,Y) + 1$ and $C(\mathbf{t}) = Y$.
	\end{itemize}
	
	The reason for having a ``$+1$'' in the 
	definition of $C(start_A)$ is so that we are guaranteed to have at least one process to begin with.
	
	From $C$, we proceed as follows: We move exactly one process from $start_A$ to $\ell_A$ by using the rule
	which increments nothing and we move all the other processes, one by one, from $start_A$ to $\ell_A$ by using the rule
	which increments $v_A$. This leads to a configuration $C'$ such that $C'(v_A) = C'(z) - 1 = C(z) - 1$.
	Because we assumed that $A(X,Y)$ is true, it follows that at $C'$, the outgoing rule from 
	$\ell_A$ is enabled. Hence, we can now move
	a process from $\ell_A$ into $end_A$, thereby covering $end_A$.
	
	($\Leftarrow$): Assume that $C$ is a simple configuration of $\TA_A[X]$ with $C(start_A) > 0$ and $C(\mathbf{t}) = Y$ such that from $C$
	it is possible to cover $end_A$. Let $\rho$ be a run from $C$ which covers $end_A$.
	By construction of $\TA_A$, it must mean that the outgoing rule from $\ell_A$ is fired at some point 
	along the run and so its guard must be enabled at some configuration $C'$ along the run.
	Note that $C'(\mathbf{t}) = C(\mathbf{t})$, since the environment variables never change their value along a run.
	
	Now, suppose $A$ is of the form $x_j | y_k$. This means that we have $X(s_j) \cdot C'(d_A) = C'(t_k)$.
	Since $X(s_j) = X(x_j), C'(t_k) = C(t_k) = Y(t_k)$, this implies that $X(x_j)$ divides $Y(t_k)$ and so $A(X,Y)$ is true.
	On the other hand, suppose $A$ is of the form $f(\bx,\by) \bowtie g(\bx,\by)$.
	Since $C'(\mathbf{t}) = C(\mathbf{t}) = Y$, this implies
	that $f(X,Y) \bowtie g(X,Y)$ and so $A(X,Y)$ is true.	
\end{proof}

\subsubsection{Second phase.} 
Let $\{\xi_i\}_{i \in I}$ be the set of subformulas of $\xi$ such that $\xi = \forall \bx \ \exists \by \bigvee_{i \in I} \xi_i$, i.e., the subformula $\xi_i$ is the disjunct corresponding to the index $i$ in the formula $\xi$. 
Let $A_i^1,\dots,A_i^{l_i}$ be the set of atomic formulas appearing in $\xi_i$.
We construct a sketch threshold automaton $\TA_{\xi_i}$ in the following manner: We take the sketches
$\TA_{A_i^1},\dots,\TA_{A_i^{l_i}}$ from the first phase and then for every $1 \le j \le l_i-1$,
we add a rule which connects $end_{A_i^j}$ to $start_{A_i^{j+1}}$, which neither increments any shared variable nor has any threshold
guards.
This is illustrated in Figure~\ref{fig:and-gadget} for the case of $l_i = 3$.

\begin{figure}
	\begin{center}
		\tikzstyle{node}=[circle,draw=black,thick,minimum size=7mm,inner sep=0.75mm,font=\normalsize]
		\tikzstyle{edgelabelabove}=[sloped, above, align= center]
		\tikzstyle{edgelabelbelow}=[sloped, below, align= center]
		\begin{tikzpicture}[->,node distance = 0.2cm,scale=0.8, every node/.style={scale=0.8}]
			\node[node] (start1) {{\small $start_{A_i^1}$}};
			\node[node, draw = none, right = of start1] (dots1) {$\dots \dots$};
			\node[node, right = of dots1] (end1) {{\small $end_{A_i^1}$}};
			\node[node, draw = none, below = of start1, xshift = 1.5cm, yshift = -0.2cm] (A1) {{\Large $A_{i}^1$}};
			\draw[dashed] (start1) ++(-0.9,1) rectangle (3.8,-1);
			
			\node[node, right = of end1, xshift = 1cm] (start2) {{\small $start_{A_i^2}$}};
			\node[node, draw = none, right = of start2] (dots2) {$\dots \dots$};
			\node[node, right = of dots2] (end2) {{\small $end_{A_i^2}$}};
			\node[node, draw = none, below = of start2, xshift = 1.5cm, yshift = -0.2cm] (A2) {{\Large $A_{i}^2$}};
			\draw[dashed] (start2) ++(-0.9,1) rectangle (9.45,-1);
			
			\node[node, right = of end2, xshift = 1cm] (startl) {{\small $start_{A_i^{3}}$}};
			\node[node, draw = none, right = of startl] (dotsl) {$\dots \dots$};
			\node[node, right = of dotsl] (endl) {{\small $end_{A_i^{3}}$}};
			\node[node, draw = none, below = of startl, xshift = 1.5cm, yshift = -0.2cm] (Al) {{\Large $A_{i}^3$}};
			\draw[dashed] (startl) ++(-0.9,1) rectangle (15.1,-1);

			\draw(end1) edge (start2);
			\draw(end2) edge (startl);
		\end{tikzpicture}
	\end{center}
\caption{Example sketch for the second phase}
\label{fig:and-gadget}
\end{figure}
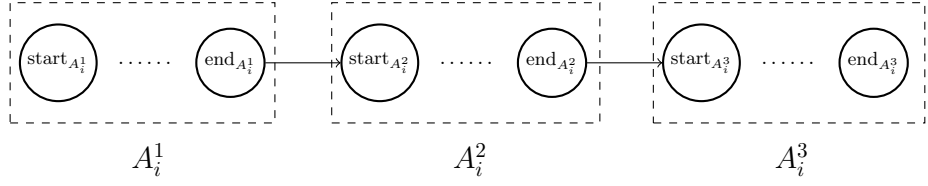


To prove a connection between the constructed gadget and the formula $\xi_i$, we first need to state
a property of the gadget. We begin with a definition.

\begin{definition}
	Let $C, C'$ be two configurations of $\TA_{\xi_i}[X]$ for some assignment $X$ and 
	let $A \in \{A^1_i,\dots,A^{l_i}_i\}$.
	We say that $C \preceq_A C'$ if $C(v) = C'(v)$ for $v \in \{v_A,d_A,\mathbf{t}\}$ 
	and $C(v) \le C'(v)$ for $v \in \{start_A,\ell_A,end_A,z\}$.
\end{definition}

By construction of $\TA_{\xi_i}$, the following \emph{monotonicity property} is clear.

\begin{proposition}[Monotonicity]
	Let $X$ be an assignment to the indeterminates and let $C \act{r} C'$ be a step in $\TA_{\xi_i}[X]$
	such that the rule $r$ belongs to $\TA_{A_i^{j}}$ for some $j$.
	Then for every $D$ such that $C \preceq_{A^j_i} D$,
	there exists a $D'$ such that $D \act{r} D'$ is a step in $\TA_{\xi_i}[X]$  and $C' \preceq_{A^j_i} D'$.
\end{proposition}

We now have the following proof which asserts the correctness of our construction.

\begin{lemma}\label{prop:second-phase}
	Let $X$ and $Y$ be assignments to the variables $\mathbf{x}$ and $\mathbf{y}$ respectively. Then $\xi_i(X,Y)$ is true if and only if there is a 
	simple configuration $C$ of $\TA_{\xi_i}[X]$ with $C(start_{A_{i}^1}) > 0$ and $C(\mathbf{t}) = Y$ such that it can cover $end_{A_{i}^{l_i}}$
\end{lemma}

\begin{proof}
	($\Rightarrow$): Suppose $\xi_i(X,Y)$ is true. Since $\xi_i = \bigwedge_{1 \le j \le l_i} A^j_i$, this means that $A^j_i(X,Y)$ is true for every $1 \le j \le l_i$. 
	By Lemma~\ref{prop:first-phase}, for every $1 \le j \le l_i$, there exists a simple configuration $C_j$
	of $\TA_{A^j_i}[X]$ with $C_j(start_{A^j_{i}}) > 0$ and $C_j(\mathbf{t}) = Y$ such that $C_j$ can cover $end_{A^j_i}$.
	For each $j$, let $C_j \act{*} C_j'' \act{r_j} C_j'$ be a shortest run from $C_j$ which covers
	$end_{A^j_i}$. By definition $C_j''(end_{A^j_i}) = 0$ and $C_j'(end_{A^j_i}) > 0$. This means
	that the (unique) outgoing rule from $\ell_{A^j_i}$ is enabled at $C_j''$ and $r_j$ is in fact, this rule.
	This also implies that the only difference between $C_j''$ and $C_j'$ is that a process has moved
	from $\ell_{A^j_i}$ to $end_{A^j_i}$. In particular, the shared variables and the environment variables
	do not change their values during this step and so the guards along the rule $r_j$ are true
	at $C_j'$ as well.
	
	Let $Z = \max\{C_j(z) : 1 \le j \le l_i\}$. Let $D_1$ be the configuration given by
	$D_1(\mathbf{t}) = Y, D_1(d_{A_i^j}) = C_j(d_{A_i^j})$ for every $A \in \{A^1_i,\dots,A^{l_i}_i\}$ which is a divisibility
	constraint, $D_1(z) = D_1(start_{A^1_i}) = Z$ and $D_1(v) = 0$ for every other $v$.
	Note that $C_1 \preceq_{A^1_i} D_1$.

	We will now show the following by induction: For any $1 \le j \le l_i$, there is a configuration $D_j$
	which is reachable from $D_1$ such that $C_j \preceq_{A^j_i} D_j, D_j(start_{A^j_i}) = Z$ and $ D_j(v_{A^k_i}) = 0$ for every $k \ge j$. 
	The base case of $j = 1$ is trivial.
	Assume that we have already shown it for some $j$ and we now want to prove it for $j+1$.
	By existence of the run $C_j \act{*} C_j'$ and because of the monotonicity property,
	there is a run $D_j \act{*} D_j'$ such that $C_j' \preceq_{A^j_i} D_j'$. 
	Since the guards of the outgoing rule from $\ell_{A^j_i}$ are enabled at $C_j'$, it follows that they are also enabled at $D_j'$.
	We now do the following: From $D_j'$, we first move all the processes at $start_{A^j_i}$ to $\ell_{A^j_i}$
	by means of the rule which increments nothing. From there we move all the processes at $\ell_{A^j_i}$
	to $end_{A^j_i}$ and then to $start_{A^{j+1}_i}$. This results in a configuration $D_{j+1}$ which satisfies the claim.
	
	By induction, this means that we can reach $D_{l_i}$ from $D_1$. By the monotonicity property,
	we can cover $end_{A_i^{l_i}}$ from $D_{l_i}$. 
	
	($\Leftarrow$): Suppose there is a 
	simple configuration $C$ of $\TA_{\xi_i}[X]$ with $C(start_{A_{i}^1}) > 0$ and $C(\mathbf{t}) = Y$ such that it can cover $end_{A_{i}^{l_i}}$. Let $C \act{*} C'$ be such a run. By construction of $\TA_{\xi_i}$, this implies
	that there must be configurations $C_1, \dots, C_{l_i}$ along this run such that at each $C_j$,
	the outgoing rule from $\ell_{A^j_i}$ must be enabled. Hence, this means that 
	if $A^j_i$ is a formula of the form $x_k |  y_{k'}$, then 
	$X(s_k) \cdot C_j(d_{A^j_i}) = C_j(t_{k'})$ and if $A^j_i$ is a formula of the form $f_j(\mathbf{x},\mathbf{y}) \bowtie g_j(\bx,\by)$, then $f_j(X(\mathbf{s}),C_j(\mathbf{t})) \bowtie g_j(X(\mathbf{s}),C_j(\mathbf{t}))$.
	Since environment variables do not change their values along a run, this implies that
	in the former case, $X(x_k) | Y(y_{k'})$ and in the latter case, $f_j(X,Y) \bowtie g_j(X,Y)$.
	Hence, $A^j_i(X,Y)$ is true for every $j$ and so $\xi_i$ is true. 
\end{proof}

\subsubsection{Third phase.} The final sketch threshold
automaton $\TA$ is constructed as follows: $\TA$ will have a copy of each of the $\TA_{\xi_i}$ and in addition
it will also have two new locations $start$ and $end$.
Then, for each index $i \in I$, $\TA$ will have two rules, one of which goes from $start$ to $start_{A_i^1}$ and the other from $end_{A_i^{l_i}}$ to $end$. Both of these rules do not increment any
variable and do not have any guards. Intuitively, these two rules correspond to choosing 
the disjunct $\xi_i$ from the formula $\xi$.
This is illustrated in Figure~\ref{fig:or-gadget} for the case when the index set $I = \{i,j,k\}$.

\begin{figure}
	\begin{center}
		\tikzstyle{node}=[circle,draw=black,thick,minimum size=7mm,inner sep=0.75mm,font=\normalsize]
		\tikzstyle{edgelabelabove}=[sloped, above, align= center]
		\tikzstyle{edgelabelbelow}=[sloped, below, align= center]
		\begin{tikzpicture}[->,node distance = 0.2cm,scale=0.8, every node/.style={scale=0.8}]
			\node[node] (start) {{$start$}};
			
			\node[node, above right = of start, xshift = 4cm, yshift = 2cm] (start1) {{\small $start_{A_i^1}$}};
			\node[node, draw = none, right = of start1] (dots1) {$\dots \dots$};
			\node[node, right = of dots1] (end1) {{\small $end_{A_i^{l_i}}$}};
			\node[node, draw = none, below = of start1, xshift = 1.5cm] (A1) {{\Large $\TA_{\xi_i}$}};
			\draw[dashed] (start1) ++(-0.9,-1) rectangle (9.05,4);
			
			\node[node, right = of start, xshift = 3.6cm] (start2) {{\small $start_{A_j^1}$}};
			\node[node, draw = none, right = of start2] (dots2) {$\dots \dots$};
			\node[node, right = of dots2] (end2) {{\small $end_{A_j^{l_j}}$}};
			\node[node, draw = none, below = of start2, xshift = 1.5cm] (A2) {{\Large $\TA_{\xi_j}$}};
			\draw[dashed] (start2) ++(-0.9,-1) rectangle (9.05,1);
			
			\node[node, below right = of start, xshift = 4cm, yshift = -2cm] (startl) {{\small $start_{A_k^{1}}$}};
			\node[node, draw = none, right = of startl] (dotsl) {$\dots \dots$};
			\node[node, right = of dotsl] (endl) {{\small $end_{A_k^{l_k}}$}};
			\node[node, draw = none, below = of startl, xshift = 1.5cm] (Al) {{\Large $\TA_{\xi_k}$}};
			\draw[dashed] (startl) ++(-0.9,1) rectangle (9.05,-4);
			
			\node[node, right = of end2, xshift = 3.6cm] (end) {$end$};
			
			\draw(start) edge (start1);
			\draw(start) edge (start2);
			\draw(start) edge (startl);		
			
			\draw(end1) edge (end);
			\draw(end2) edge (end);
			\draw(endl) edge (end);		
		\end{tikzpicture}
	\end{center}
\caption{Example sketch for the third phase}
\label{fig:or-gadget}
\end{figure}
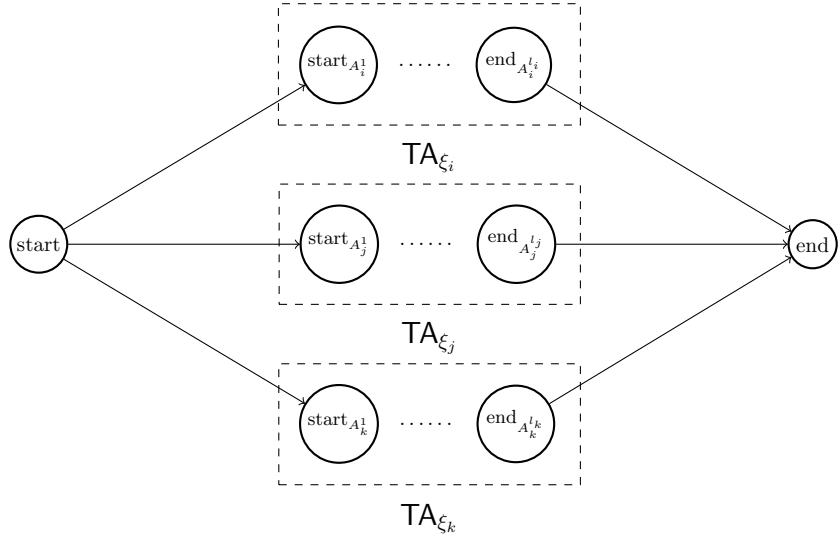


Setting the initial set of locations of $\TA$ to be $\{start\}$, we have the following lemma.

\begin{lemma}\label{prop:third-phase}
	Let $X$ and $Y$ be assignments to the variables $\mathbf{x}$ and $\mathbf{y}$ respectively. Then $\xi(X,Y)$ is true if and only if
	some initial configuration $C$ with $C(\mathbf{t}) = Y$ can cover the location $end$ in $\TA[X]$.	
\end{lemma}

\begin{proof}
	($\Rightarrow$): Suppose $\xi(X,Y)$ is true. Then $\xi_i(X,Y)$ is true for some $i$. By Lemma~\ref{prop:second-phase},
	there exists a simple configuration $C$ of $\TA_{\xi_i}[X]$ with $C(start_{A^1_i}) > 0$ and 
	$C(\mathbf{t}) = Y$ which can cover $end_{A^{l_i}_i}$. Consider the initial configuration $D$ in $\TA$
	which is the same as $C$ except that $D(start_{A^1_i}) = 0$ and $D(start) = C(start_{A^1_i})$.
	By construction of $\TA[X]$, we can make $D$ reach $C$. Since we can cover $end_{A^{l_i}_i}$ from $C$ in $\TA_{\xi_i}[X]$,
	we can also cover it in $\TA[X]$. Once we can cover $end_{A^{l_i}_i}$, we can also cover $end$.
	
	($\Leftarrow$): Suppose there is some initial configuration $C$ with $C(\mathbf{t}) = Y$ which can cover the location $end$ in $\TA[X]$. Let $C \act{*} C'$ be such a run covering $end$. By construction, there must be an index $i \in I$ and 
	configurations $C_1,C_2,\dots,C_{l_i}$ along this run such that at each $C_j$,
	the outgoing rule from $\ell_{A^j_i}$ is enabled.
	Similar to the argument from Lemma~\ref{prop:second-phase}, we can then show that each $A^j_i(X,Y)$ is true
	and so $\xi_i(X,Y)$ is true, which implies that $\xi(X,Y)$ is true. 
\end{proof}

It then follows that $\forall \bx \ \exists \by \ \xi(\bx,\by)$ is true
if and only if for every assignment $X$ of the indeterminates of $\TA$, there exists an initial configuration $C$ such 
that $C$ can cover $end$ in $\TA[X]$. Hence, the formula $\forall \bx \ \exists \by \ \xi(\bx,\by)$ is false
if and only if there exists an assignment $X$ of the indeterminates of $\TA$ such that $\TA[X]$ does not cover $end$. Since $\TA$ is acyclic, we then get the 
following theorem.
\begin{theorem}
	The non-negative coefficient synthesis problem for threshold automata is undecidable, even for acyclic threshold automata.
\end{theorem}

\begin{example}
	We illustrate the above reduction on an example. Suppose we have the formula
	\begin{equation}\label{eq:example}
		\forall x_1, x_2 \ \exists y_1, y_2 \ (x_1 | y_1) \lor (x_2 | y_1 \land x_1 = 2x_2 + y_2) 
	\end{equation}
	
	Let $A, \ B,$ and $C$ denote the sub-formulas $x_1 | y_1, \ x_2 | y_1$ and $x_1 = 2x_2 + y_2$ respectively.
	For the formula~\ref{eq:example}, our reduction produces the sketch given in Figure~\ref{fig:example}.
	
	\begin{figure}
		\begin{center}
			\tikzstyle{node}=[circle,draw=black,thick,minimum size=7mm,inner sep=0.75mm,font=\normalsize]
			\tikzstyle{edgelabelabove}=[sloped, above, align= center]
			\tikzstyle{edgelabelbelow}=[sloped, below, align= center]
			\begin{tikzpicture}[->,node distance = 1.8cm,scale=0.8, every node/.style={scale=0.8},initial text = {}]
				\node[node,initial] (start) {{$start$}};
				
				\node[node, above right = of start, xshift=0.7cm] (start1) {{\small $start_A$}};
				\node[node, right = of start1] (ell1) {$\ell_A$};
				\node[node, right = of ell1, xshift=2.2cm] (end1) {{\small $end_A$}};

				\node[node, right = of start, xshift=-0.3cm] (start2) {{\small $start_B$}};
				\node[node, right = of start2] (ell2) {$\ell_B$};
				\node[node, right = of ell2, xshift=2.2cm] (end2) {{\small $end_B$}};
				
				\node[node, below right = of start, xshift=0.7cm] (start4) {{\small $start_C$}};
				\node[node, right = of start4] (ell4) {$\ell_C$};
				\node[node, right = of ell4, xshift=2.2cm] (end4) {{\small $end_C$}};
				
				\node[node, right = of end2, xshift=-1cm] (end) {$end$};
				
				\draw(start) edge (start1);
				\draw(start) edge (start2);
				
				
				\draw(start1) edge[edgelabelabove, bend left = 15] node{$\cpp{v_A}$} (ell1);
				\draw(start1) edge[edgelabelbelow, bend right = 15] node{} (ell1);
				\draw(ell1) edge[edgelabelabove] node{$v_A = s_1 \cdot d_A \ \land \ v_A = t_1$}(end1);
				
				\draw(start2) edge[edgelabelabove, bend left = 15] node{$\cpp{v_B}$} (ell2);
				\draw(start2) edge[edgelabelbelow, bend right = 15] node{} (ell2);
				\draw(ell2) edge[edgelabelabove] node{$v_B = s_2 \cdot d_B \ \land \ v_B = t_1$} (end2);
				
				\draw(end2) edge (start4);
								
				\draw(start4) edge[edgelabelabove, bend right = 15] node[below]{$\cpp{v_C}$} (ell4);
				\draw(start4) edge[edgelabelbelow, bend left = 15] node{} (ell4);
				\draw(ell4) edge[edgelabelabove] node{$v_C = s_1 \ \land \ v_C = 2s_2+t_2$} (end4);

				\draw(end1) edge (end);
				\draw(end4) edge (end);
				
			\end{tikzpicture}
		\end{center}
	\caption{Sketch for formula~\ref{eq:example}}
	\label{fig:example}
	\end{figure}
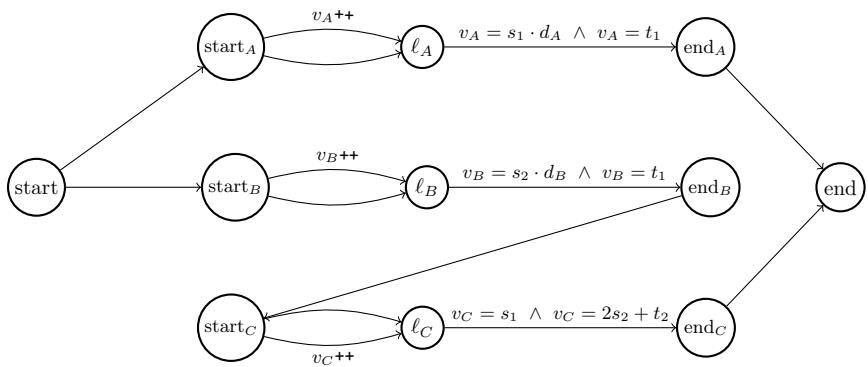

Here $s_1, s_2$ are indeterminates corresponding to $x_1, x_2$ and $t_1,t_2$ are environment variables corresponding to $y_1,y_2$. Notice that the formula is true,
because if $x_1$ is assigned the value $a$ and $x_2$ is assigned the value $b$, then we can always set $y_1$ to $a$ and $y_2$ to $b$, which will always make the first disjunct true. 
Similarly, in the sketch threshold automaton, if $\mu$ is any assignment to the indeterminates, then by letting $C$ be the (unique) initial configuration such that $C(t_1) = \mu(s_1), \ C(t_2) = \mu(s_2), \ C(z) = \mu(s_1) + 1$ and $C(d_A) = C(d_B) = 1$, we can cover $end_A$ from $C$ and so we can also cover $end$ from $C$.
\end{example}

\subsection{Wrapping up}

We can now reduce the non-negative coefficient synthesis problem
to the coefficient synthesis problem, thereby proving Theorem~\ref{th:main-th}.
%

\paragraph*{Proof of Theorem~\ref{th:main-th}.}  
	Let $(\Env,\TA,end)$ be an instance of the non-negative coefficient synthesis problem. 
	Without loss of generality, we can assume that $\TA$ is acyclic and
	has only a single initial location $start$. This is because we have shown
	earlier that the non-negative coefficient synthesis problem is already undecidable for inputs satisfying this property.
	
	Let $X$ be the set of indeterminates of $\TA$. We now add a new location $begin$ and a new shared variable $check$. $check$
	will have the invariant that it will never be incremented by any of the rules. Now, from $begin$ we add $|X|+1$
	rules as follows: First, we add a rule from $begin$ to $start$ which neither increments any variable nor has any guards.
	Then for each indeterminate $x \in X$, we add a rule from $begin$ to $end$ which has the guard $check > x$. 
	Notice that since $check$ is never incremented, it will always have the value 0 and so the guard $check > x$ will be true if and only if
	$x$ takes a negative value. Finally, we set the new initial location to be $begin$ and we let this new sketch threshold automaton be $\TA'$. Notice that $\TA'$ is acyclic.
	
	We will now prove that
	$(\Env,\TA,end)$ is a yes instance of the non-negative coefficient synthesis problem if and only if $(\Env,\TA',end)$ is a yes instance of the coefficient synthesis problem.
	
	Notice that if $\mu$ is an assignment to $X$ such that $\mu(x) < 0$ for some $x \in X$, then it is possible to move a process
	from $begin$ to $end$. Hence, if $\mu$ assigns a negative value to some indeterminate, then there is at least one run from some initial configuration in $\TA'[\mu]$ which covers $end$.
	Hence, if no initial configuration of $\TA'[\mu]$ can cover $end$, then $\mu$ has to be a non-negative assignment. But then
	it is easy to see that no initial configuration of $\TA[\mu]$ can cover $end$ as well. 
	
	Similarly, suppose $\mu$ is a non-negative assignment such that $\TA[\mu]$ does not cover $end$. Then,
	since $\mu$ is non-negative, it is clear that the only rule which can be fired from $begin$ is the rule which moves a process from $begin$ to $start$.
	Hence, it is then clear that $\TA'[\mu]$ also cannot cover $end$.

\section{Complexity of Bounded Coefficient Synthesis}\label{sec:bound-synth}
We have seen that the coefficient synthesis problem for threshold automata is undecidable. 
Since the reachability relation for threshold automata is decidable~\cite[Corollary 1]{Complexity}, the source of undecidability
stems from the unboundedness of the values of the indeterminates needed to satisfy the given
specification. 
Intuitively, this is undesirable, because we would ideally like our indeterminates to not take very big values.
This leads to the following \emph{bounded coefficient synthesis} problem, where we are given a 
tuple $(\Env,\TA,\ell)$ as in the coefficient synthesis problem, and in addition, we are given an interval $[A,B]$ with $A, B \in \zn$.
We are then asked to check if there is an integer assignment $\mu$ to the indeterminates such that $A \le \mu(x) \le B$ for every indeterminate $x$
and $\TA[\mu]$ does not cover $\ell$. \footnote{We can also allow a separate interval for each indeterminate. All the results in this section would be applicable to that case as well.}

In this section, we will revisit the known upper and lower bounds for this problem and present
different algorithms and constructions for both the upper and the lower bounds. Then, we will consider a 
special case based on existing benchmarks from the literature and provide better bounds for this case.
We begin by concentrating on the upper bound.

\subsection{Upper Bound for Bounded Coefficient Synthesis}

To start with, the following upper bound is known for the bounded coefficient synthesis problem.

\begin{theorem}\label{thm:bounded-sythesis}~\cite[Theorem 7]{Complexity}
	Bounded coefficient synthesis is in $\Sigma_2^p$.
	\footnote{In~\cite{Complexity} this $\Sigma_2^p$ upper bound was even proven for a broader class of specifications at the cost of restricting the type of automata to those satisfying a property called \emph{multiplicativity}.}
\end{theorem}

To explain this result, we need to state a few results regarding threshold automata.
It is known that given a threshold automaton $\TA$, we can construct an exponential-sized formula $\xi$ in \emph{existential Presburger arithmetic} which characterizes precisely the reachability relation of $\TA$~\cite[Theorem 5]{Complexity}.
Since the exponential dependence only comes from exponentially many disjuncts, given two configurations $\sigma$
and $\sigma'$ of $\TA$, we can decide if $\sigma \act{*} \sigma'$ in \NP \ by simply \emph{guessing} and constructing one of the disjuncts of $\xi$ and then
using the fact that the existential theory of PA can be decided in \NP.

Note that once we have this result, checking whether a given location $\ell$ cannot be covered is in \coNP: we quantify universally over all pairs of configurations $\sigma$ and $\sigma'$ and verify that
if $\sigma$ is an initial configuration and $\sigma'(\ell) > 0$, then $\sigma$ cannot reach $\sigma'$.
With this observation, bounded synthesis can then be easily seen to be in $\Sigma_2^p$: simply guess
a value for each indeterminate within the given range, plug-in the guessed values into the sketch and then run the \coNP \ decision procedure.
%

This argument implies that if we were able to reduce the size of the formula characterizing the reachability
relation, then it would automatically lead to shorter formulas for the bounded coefficient synthesis problem as well.
Within this context, our main contribution is the following theorem.

\begin{theorem}\label{thm:quad-size}
	Given a pair $(\Env,\TA)$, we can construct in polynomial time, a polynomial-sized formula $\phi$ of existential
	Presburger arithmetic such that $\sigma \act{*} \sigma'$ is true in $\TA$ if and only if $\phi(\sigma,\sigma')$ is true.
\end{theorem}

Hence, this result improves on the previous upper bound of formulas characterizing the
reachability relation. In the rest of this subsection, we will prove Theorem~\ref{thm:quad-size}. 
The formula that we give is similar to the formula given in~\cite[Theorem 4]{CADE2005} for 
the Parikh image of a CFG.
Further since our formula differs only slightly 
from the formula given in~\cite{Complexity}, we will present the formula in a similar manner as the one 
given in~\cite{Complexity}. 

\paragraph*{Existential PA formula for reachability.} Fix a threshold automaton $\TA = (\local,\initlocal,\globset,\ruleset)$
over an environment $\Env$. First, we will construct a polynomial-sized formula for a \emph{restricted form of reachability} between configurations. 
We now proceed to describe this restricted reachability relation.

Let $g$ be a guard of $\TA$ of the form $b \cdot x \bowtie a_0 + a_1 \cdot p_1 + \ldots + a_{|\paraset|} \cdot p_{|\paraset|}$. Without loss of generality, we can assume that $\bowtie \ \in \{\ge, <\}$, since any other
guard could be written as a conjunction of guards of this form.
Then, $g$ is called a \emph{rise guard} (resp. \emph{fall guard}) if $\bowtie \ =  \ \ge$ 
(resp. $\bowtie \ = \ <$).

Given a configuration $\sigma$, the \emph{context} of $\sigma$, denoted by $\omega(\sigma)$ is the set of all rise guards
that evaluate to true and the set of all fall guards that evaluate to false in $\sigma$. 
We say that a run $\sigma \act{*} \sigma'$ is \emph{steady} if the set of all configurations visited along this run
have the same context. Since shared variables are only incremented and environment variables never change values along a run, it follows that if a rise guard
becomes true at some point along a run, then it stays true throughout the run. Similarly, if 
a fall guard becomes false at some point along a run, then it stays false throughout the run.
It then follows that a run is steady if and only if
the first and the last configurations of this run have the same context.
We will now construct a formula $\phi_{steady}$ with $(2|\local| + 2|\globset| + 2|\paraset|)$ free variables
such that $\phi_{steady}(\sigma,\sigma')$ is true if and only if $\sigma$ and $\sigma'$ have the same context
and $\sigma$ can reach $\sigma'$.

\paragraph*{The formula $\phi_{steady}$. } For every rule $r \in \ruleset$, we will introduce a variable $x_r$, which will
intuitively denote the number of times the rule $r$ is fired during the (supposed) run from $\sigma$ to $\sigma'$. 
Let $X = \{x_r\}_{r \in \ruleset}$.
The formula $\phi_{steady}$ is obtained by a conjunction of various subformulas, described as follows.

\paragraph{Subformula 1.} $\sigma$ and $\sigma'$ must have the same context, the same number of processes and the same values over the environment variables which must also satisfy
the resilience condition.
$$\phi_\mathit{base}(\sigma,\sigma') \ \equiv \;  \sigma.\param = \sigma'.\param \; \wedge \;
\ResCond(\sigma.\param) \; \wedge \; \syssize(\sigma.\param) = 
\syssize(\sigma'.\param) \; \wedge \;
\statectx(\sigma)=\statectx(\sigma')$$

\paragraph{Subformula 2.} For a location $\ell\in\local$, 
let $out^\ell_1,\ldots,out^\ell_{a_\ell}$ be all the outgoing rules from~$\ell$ and let $in^\ell_1,\dots,in^\ell_{b_\ell}$ be all the 
incoming rules to~$\ell$. The number of processes in $\ell$ after the run must be the initial number, plus the incoming processes, minus the outgoing processes. Hence, we have
$$	\phi_\local(\sigma,\sigma',X) \ \equiv \quad
\bigwedge_{\ell \in \local} \ \left(\sum_{i=1}^{b_\ell} x_{in^\ell_i} - \sum_{j=1}^{a_\ell} x_{out^\ell_j} =
\sigma'.\counters(\ell) - \sigma.\counters(\ell)\right)
$$
\paragraph{Subformula 3.} Similarly, for the shared variables we must have
$$  \phi_\globset(\sigma,\sigma',X)  \ \equiv \quad
\bigwedge_{z\in\globset} \ \left(\sum_{r\in\ruleset} (x_{r} \,\cdot\, r.\update[z]) = \sigma'.\vars[z] - \sigma.\vars[z]\right)
$$
\paragraph{Subformula 4.} Since we are searching for a steady run between $\sigma$ and $\sigma'$, for a rule to be fired
along this run, it is necessary that its guards are true in $\sigma$.
$$\phi_\ruleset(\sigma,X) \ \equiv \quad \bigwedge_{r \in \ruleset} \ x_r > 0 \ \Rightarrow \ (\sigma\models r.\varphi)$$
\paragraph{Subformula 5.} Finally, we introduce a collection of variables $Y = \{y_r\}_{r \in \ruleset}$, which intuitively captures the following observation: If a rule $r$ is fired in a run between
$\sigma$ and $\sigma'$, then either $\sigma(r.\fromstate) > 0$ or
there must be a rule $r'$ which is fired before $r$ such that $r'.\tostate = r.\fromstate$.
The following formula enforces this condition using the variables $Y$.

$$  \phi_{\textit{appl}}(\sigma,X,Y) \ \equiv \quad
\bigwedge_{r \in \ruleset} \; x_r > 0 \ \Rightarrow \ \phi_\mathit{chain}^r(\sigma,X,Y)
$$
where
\begin{multline*}
	\phi_\mathit{chain}^r(\sigma,X,Y) \ \equiv \;
	\left(\sigma(r.\fromstate) > 0 \ \land \ y_r = 1\right) \ 
	 \lor \ \left(\bigvee_{r' \in Pre(r)} (x_{r'} > 0 \ \land \ y_r = y_{r'} + 1) \right)	
\end{multline*}
and $Pre(r) = \{r' : r'.\tostate = r.\fromstate\}$.

\paragraph{Combining the steps. } We then define $\phi_\mathit{steady}(\sigma,\sigma')$ as
\begin{multline*}\label{eq:kirchoff}
	\phi_\mathit{steady}(\sigma,\sigma') 
	\ \equiv \ \phi_\mathit{base}(\sigma,\sigma') \ \land \ \exists X, Y \ 
	\phi_\local(\sigma,\sigma',X)  \ \land \ 
	\phi_\globset(\sigma,\sigma',X)  \ \land \\
	\phi_\ruleset(\sigma,X)  \ \land  \
	\phi_{\textit{appl}}(\sigma,X,Y)	
\end{multline*}
\noindent Notice that the size of $\phi_{steady}$ is polynomial in the size of the given threshold automaton and the environment. We now have the following theorem.

\begin{theorem}\label{th:kirchoffpath}
	Let~$\TA$ be a threshold automaton and let $\sigma, \sigma'$ be two configurations.
	Formula $\phi_\mathit{steady}(\sigma,\sigma')$ is satisfiable if and only if 
	there is a steady run between $\sigma$ and $\sigma'$.
\end{theorem}

Before proving this theorem, let us see how it implies Theorem~\ref{thm:quad-size}. 
If the underlying threshold automaton has $K$ guards, then given \emph{any} formula $\theta$ which characterizes the steady reachability relation, the authors of~\cite{Complexity} come up with a formula $\theta'$ whose size is at most
$O(K) \times |\theta|$ such that $\theta'$ characterizes the reachability relation, i.e., $\theta'(\sigma, \sigma')$
is true if and only if $\sigma$ can reach $\sigma'$.
Using their procedure, Theorem~\ref{th:kirchoffpath} then implies that we have a polynomial sized formula for the reachability relation and proves Theorem~\ref{thm:quad-size}.

Hence, all that is left is to prove Theorem~\ref{th:kirchoffpath}, which we do now.

\paragraph*{Proof of Theorem~\ref{th:kirchoffpath}.} 
	In~\cite{Complexity}, the authors present a formula $\xi_{steady}$ such that $\xi_{steady}(\sigma,\sigma')$ is 
	satisfiable if and only if there is a steady run between $\sigma$ and $\sigma'$.
	Hence to prove the theorem it suffices to show that $\phi_{steady}(\sigma,\sigma')$ is satisfiable if and only if $\xi_{steady}(\sigma,\sigma')$ is satisfiable.
	
	To do this, we first explain the formula $\xi_{steady}(\sigma,\sigma')$. Similar to $\phi_{steady}$,
	$\xi_{steady}(\sigma,\sigma')$ is of the form 
	\begin{eqnarray*}\label{eq:kirchoff}
		\xi_\mathit{base}(\sigma,\sigma') \ \land  \ \exists X \ \xi_\local(\sigma,\sigma',X) \ \land \
		\xi_\globset(\sigma,\sigma',X) \  \land \
		\xi_\ruleset(\sigma,X)  \ \land \  
		\xi_{\textit{appl}}(\sigma,X) 
	\end{eqnarray*}
	where $\xi_q$ is the same as $\phi_q$ for every $q \in \{base, \local, \globset, \ruleset\}$
	and $\xi_{appl}(\sigma,X)$ is defined as follows:
	$$  \xi_{\textit{appl}}(\sigma,X) \ \equiv \quad
	\bigwedge_{r \in \ruleset} \; \left( x_r > 0 \ \Rightarrow \ \bigvee_{S=\{r_1,r_2,\dots,r_s\} \subseteq \ruleset} \xi_\mathit{chain}^r(S,\sigma,X)\right)
	$$
	where
	\begin{multline*}
		\xi_\mathit{chain}^r(S,\sigma,X) \ \equiv \
		\bigwedge_{1 \le i \le s} x_{r_i} > 0 \ \land \  \sigma.\counters(r_1.\fromstate) > 0 \ \land \
		\bigwedge_{1 < i \le s} r_{i-1}.\tostate = r_i.\fromstate \; \land \; r_s = r	
	\end{multline*}
	
	Now, let $Z$ be any non-negative assignment to the variables of $X$, i.e., we assign to the variable $x_r$, the natural number $z_r$.
	Suppose we show that $\xi_{appl}(\sigma,Z)$ is true if and only if there exists an non-negative assignment $Z'$ to the variables of $Y$ such that $\phi_{appl}(\sigma,Z,Z')$ is true.
	Then notice that our proof would be complete. This is what we proceed to do now.
	
	($\Rightarrow$): Suppose $\xi_{appl}(\sigma,Z)$ is true. Hence, for every rule $r$ such that $z_r > 0$, there is a subset $S_r = \{t_1,\dots,t_{|S_r|}\} \subseteq \ruleset$
	such that for every $i$, $z_{t_i} > 0,$ $\sigma(t_1.\fromstate) > 0,$ $t_{|S_r|} = r$ and for every $i > 1$, $t_{i-1}.\tostate = t_i.\fromstate$. We then construct our assignment $Z'$ to the variables of $Y$ by induction on the size of $|S_r|$.

	First, if $z_r = 0$, then we set $z'_r = 0$ as well.
	Then, if $z_r > 0$ and $|S_r| = 1$, we set $z'_r = 1$.
	Finally, if $z_r > 0$ and $|S_r| > 1$, then let $t$ be the penultimate rule of 
	$S_r$. Notice that we can assume that $|S_t| < |S_r|$, because the set $S_r \setminus \{r\}$ satisfies all the conditions needed
	for $S_t$. Furthermore $z_t > 0$ and 
	so by the induction hypothesis, we have already defined $z'_t$. We then set $z'_r = z'_t + 1$. 

	Now, suppose $z_r > 0$ for some rule $r$. Then notice that $S_r$ exists and is non-empty.
	If $|S_r| = 1$, then 
	$\sigma(r.\fromstate) > 0$ and we have set $z'_r = 1$. 
	On the other hand if $|S_r| > 1$, then 
	letting $t$ be the penultimate rule of $S_r$, we have that $t.\tostate = r.\fromstate, 
	z_t > 0$ and $z'_r = z'_t + 1$. It then follows that $\phi^r_{chain}(\sigma,Z,Z')$ is true
	for every $r$ with $z_r > 0$. This then implies that $\phi_{appl}(\sigma,Z,Z')$ is also true.
	
	($\Leftarrow$): Now, suppose there exists a non-negative assignment $Z'$ to the variables of $Y$ such that $\phi_{appl}(\sigma,Z,Z')$ is true.
	We need to show that $\xi_{appl}(\sigma,Z)$ is true. To do this, we need to show that if $z_r > 0$ then there is a 
	subset $S_r = \{t_1,\dots,t_{|S_r|}\} \subseteq \ruleset$
	such that for every $i$, $z_{t_i} > 0$, $\sigma(t_1.\fromstate) > 0$, $t_{|S_r|} = r$ and for every $i > 1$, $t_{i-1}.\tostate = t_i.\fromstate$.
	We do this by induction on the value of $z'_r$. 
	
	First, note that $z'_r$ cannot be 0. Indeed if $z'_r = 0$, then since $z_r > 0$, by definition of $\phi_{appl}(\sigma,Z,Z')$
	it must be the case that there must be a rule $t$ such that $z'_r = z'_t + 1$. But this would
	mean that $z'_t = -1$, contradicting the fact that $Z'$ is a non-negative assignment. Hence, $z'_r > 0$.
	
	Suppose $z'_r = 1$. Then by definition of $\phi_{appl}(\sigma,Z,Z')$, it must be the case that $\sigma(r.\fromstate) > 0$.
	Indeed, if $\sigma(r.\fromstate) = 0$, then since $z_r > 0$, there must be a rule $t$ such that $z_t > 0$ and
	$z'_r = z'_t + 1$. This means that $z'_t = 0$ and $z_t > 0$, which, as we have shown in the previous paragraph, cannot happen. 
	Hence, $\sigma(r.\fromstate) > 0$ and so we can simply set $S_r$ to be $\{r\}$.
	
	Suppose $z'_r > 1$. By definition of $\phi_{appl}(\sigma,Z,Z')$, it must be the case that there is a rule $t$ such that
	$t.\tostate = r.\fromstate, z_t > 0$ and $z'_r = z'_t + 1$. 
	By induction hypothesis, we have already constructed the set $S_t$ for $t$. We then set $S_r = S_t \cup \{r\}$. It is then easy to verify that the constructed sets $S_r$ for each rule $r$ satisfy the desired property. This completes the proof of Theorem~\ref{th:kirchoffpath}.

\subsection{Lower Bound for Bounded Coefficient Synthesis}

In this subsection, we prove the following lower bound.

\begin{theorem}\label{thm:bounded-synth-hard}
	Bounded coefficient synthesis is $\Sigma_2^p$-hard, even for acyclic threshold automata.
\end{theorem}

Before, we move on to the proof of this theorem, we make a remark.

\begin{remark}
	A $\Sigma_2^p$ lower bound was also proven in~\cite[Theorem 8]{Complexity} for bounded coefficient synthesis
	against arbitrary violations from the $\ELTLFT$ logic. In particular, that reduction did not use a coverability violation.
	By modifying that reduction, it is possible to give a $\Sigma_2^p$ lower bound for coverability violations as well. However,
	here we give a self-contained proof of that same result.
\end{remark}

We prove this theorem by giving a reduction from the $\Sigma_2$-SAT problem.
Here we are given a Boolean formula of the form
\begin{equation}\label{eq:bool}
	\exists x_1, \dots, x_n \ \forall y_1,\dots,y_m \
	\bigvee_{1 \le j \le k} D_j
\end{equation}
where each $D_j$ is a conjunction of literals, i.e.,
it is a conjunction of entries from $\{x_1,\dots,x_n,y_1,\dots,y_m\} \cup \{\overline{x_1},\dots,\overline{x_n},\overline{y_1},\dots,\overline{y_m}\}$.
Given such a formula, the task is then to decide whether it is valid. This problem
is known to be $\Sigma_2^p$-hard~\cite[Corollary 6]{Wrathall76}.

Now, assume that we are given a formula $\xi$ of the form~\ref{eq:bool}. 
Let $\bx$ and $\by$ denote the sets of variables $x_1,\dots,x_n$ and $y_1,\dots,y_m$ respectively.
Our reduction will now proceed in two stages. We describe
the first stage now.

\paragraph*{First stage: The environment.} We begin by defining the environment  $\Env = (\paraset, \ResCond, \syssize)$. We will 
have $m$ environment variables $t_1,\dots,t_m$, with each $t_i$ intuitively corresponding to the variable $y_i$ in the formula $\xi$.
We will also have another environment variable $z$, which will intuitively denote the total number of participating processes.

The resilience condition $\ResCond$ will be the trivial condition $true$. The linear function $N : \ResCond \to \nn$
is taken to be $N(\Pi) = z$.  Hence, the total number of processes executing the threshold automaton will be $z$. 

\paragraph*{Second stage: The indeterminates and the sketch.} 
For each variable $x_i$ in the formula $\xi$, we will have an indeterminate $s_i$. Before we proceed with the description of the sketch, we set up some notation.

We let $\mathbf{s}$ denote the set of indeterminates $s_1,\dots,s_n$ and 
$\mathbf{t}$ denote the set of environment variables $t_1,\dots,t_m$. For any literal $\ell$,
\begin{itemize}
	\item If $\ell = x_i$ (resp. $y_i$) for some $i$, let $F(\ell)$ be the term $1-s_i$ (resp. $1-t_i$).
	\item If $\ell = \overline{x_i}$ (resp. $\overline{y_i}$) for some $i$, let $F(\ell)$ be the term $s_i$ (resp. $t_i$).
\end{itemize}

Given any valuation $X$ to the Boolean variables $x_1,\dots,x_n$, 
let $B(X)$ be the assignment to the indeterminates $s_1,\dots,s_n$ which assigns each $s_i$ the value 1 if $X(x_i)$ is true and the value 0 if $X(x_i)$ is false.
Given any valuation $Y$ to $y_1,\dots,y_m$, let $V(Y)$ be the \emph{set} of assignments
to the environment variables $t_1,\dots,t_m$ which assigns each $t_i$ any strictly positive value
if $Y(y_i)$ is true and the value 0 if $Y(y_i)$ is false. Note that $B(X)$ is a single assignment
whereas $V(Y)$ is a set of assignments.

Now, let $D_j$ be any conjunct of the formula $\xi$.
Corresponding to $D_j$, let $E_j$ be the term defined by $E_j = \sum_{\ell \in D_j} F(\ell)$.
From the definition of $E_j$, we have the following lemma.

\begin{lemma}\label{lem:first-phase}
	Let $X$ and $Y$ be any assignments to the variables $\bx$ and $\by$ respectively.
	Then $D_j(X,Y)$ is true if and only if for any assignments $S,T$ with $S = B(X)$ and $T \in V(Y)$,
	$E_j(S,T) \le 0$.
\end{lemma}

\begin{proof}
	($\Rightarrow$): Suppose $D_j(X,Y)$ is true. This means
	that any literal $\ell$ that appears in $D_j$ is set to true by the assignments $X$ and $Y$.
	Let $S = B(X)$ and $T \in V(Y)$ be some two assignments. Note that $E_j$ is a sum of terms
	of the form $F(\ell)$ with $\ell \in D_j$. 
	
	Now, pick any $\ell \in D_j$. By assumption, $\ell$ is set to true by $X$ and $Y$. By definition of $F(\ell)$, $S$ and $T$, it follows that $F(\ell)$ evaluates to a value that is at most 0 under $S$ and $T$. Hence, it follows that $E_j(S,T) \le 0$.
	
	($\Leftarrow$): Suppose for any pair of assignments $S, T$ with  $S = B(X)$ and $T \in V(Y)$ we have
	that $E_j(S,T) \le 0$.
	Pick the assignment $T' \in V(Y)$ such that $T'(t_i) = 1$ if $Y(y_i)$ is true
	and $T'(t_i) = 0$ otherwise. 
	By assumption $E_j(S,T') \le 0$. Note that $E_j$ is a sum of terms of the form
	$F(\ell)$ with $\ell \in D_j$. 
	
	Now, pick any $\ell \in D_j$. By construction of $F(\ell)$, $S$ and $T'$ it follows that
	its value cannot go strictly below the value 0 under the assignments $S$ and $T'$.
	This combined with the fact that $E_j(S,T') \le 0$ implies that $F(\ell)$ evaluates
	to the value 0 under $S$ and $T'$. By definition of $F(\ell)$,
	this immediately implies that $\ell$ is set to true by the assignments $X$ and $Y$.
	Since $\ell$ was any arbitrary literal from $D_j$, it follows that $D_j(X,Y)$ is true.
\end{proof}

We will now construct the desired sketch $\TA$. Recall that the formula $\xi$ has $k$ disjuncts
$D_1, D_2 \dots, D_k$. For each $D_j$, the sketch
will have one location $start_j$. Furthermore, $\TA$ has another location $start_0$ 
and a single shared variable $v$. The rules of the sketch $\TA$ are now given as follows: 
For every $1 \le j \le k$, there is a rule between $start_{j-1}$ and $start_j$ with
the threshold guard $v < E_j$.

Let the initial set of locations of $\TA$ be $\{start_0\}$. It is easy to see that the constructed sketch is acyclic. 
%
We now show the following lemma which proves the correctness of our construction.

\begin{lemma}\label{lem:second-phase}
	Let $X$ and $Y$ be any assignments to the variables $\bx$ and $\by$ respectively.
	Then $D_j(X,Y)$ is true for some $1 \le j \le k$ if and only if 
	for any assignments $S, T$ with $S = B(X)$ and $T \in V(Y)$, 
	any initial configuration $C$ with $C(\mathbf{t}) = T$ cannot cover $start_k$ in $\TA[S]$.
\end{lemma}

\begin{proof}
	($\Rightarrow$): Suppose there exists some $1 \le j \le k$ such that $D_j(X,Y)$ is true. 
	Let $S = B(X)$ and $T \in V(Y)$ be two assignments to $\mathbf{s}$ and $\mathbf{t}$ respectively.
	By Lemma~\ref{lem:first-phase}, it follows that $E_j(S,T) \le 0$. By construction
	of $\TA[S]$, the shared variable $v$ is never incremented and so it will always have the value
	0 on any run starting from any initial configuration. Since $E_j(S,T) \le 0$,
	it then follows that in any run from any initial configuration, the rule
	between $start_{j-1}$ and $start_{j}$ is never enabled and hence can never be fired. 
	This then implies that $start_{j}$ is not coverable from any initial configuration $C$ with $C(\mathbf{t}) = T$.
	By the structure of $\TA$, it follows that $start_{k}$ cannot also be covered.
	
	($\Leftarrow$): Suppose for any assignments $S, T$ with $S = B(X)$ and $T \in V(Y)$,
	any initial configuration $C$ with $C(\mathbf{t}) = T$ cannot cover $start_k$ in $\TA[S]$.
	Let $T'$ be the assignment in $V(Y)$ such that $T'(t_i) = 1$ if $Y(y_i)$ is true
	and $T'(t_i) = 0$ if $Y(y_i)$ is false.
	By assumption, no initial configuration with $C(\mathbf{t}) = T'$ must be able
	to cover $start_{k}$ in $\TA[S]$. Let $j$ be the smallest index such that no initial configuration $C$ with $C(\mathbf{t}) = T'$ can cover $start_j$ in $\TA[S]$. 
	Note that $j > 0$, as $start_0$ is an initial location.
	
%
 	Now, consider the rule between $start_{j-1}$ and $start_j$. 
	By assumption on $start_j$, it must be the case that $v < E_j(S,T')$ cannot be true, as
	otherwise, $start_j$ can be covered from $C$.
	Since the shared variable $v$ is never incremented, it follows that $E_j(S,T') \le 0$. 
	In the second part of the proof of Lemma~\ref{lem:first-phase}, we have
	already proven that $E_j(S,T') \le 0$ (for this specific assignment $T'$) already
	implies that $D_j(X,Y)$ is true. This completes the proof. 
\end{proof}

Notice that any assignment $S$ to the indeterminates $\mathbf{s}$ within the range $[0,1]$
is of the form $B(X)$ for some assignment
$X$ to the variables $\bx$. Similarly, any assignment $T$ to the environment variables $\mathbf{t}$ belongs
to $V(Y)$ for some assignment $Y$ to the variables $\by$.
Hence, by Lemma~\ref{lem:second-phase}, we get that 
there exists an assignment $X$ to $\bx$ such that for all 
assignments $Y$ to $\by$, at least one $D_j(X,Y)$ is true if and only if there 
exists an assignment $S$ to the indeterminates $\mathbf{s}$ within the range $[0,1]$
such that any initial configuration cannot cover $start_k$ in $\TA[S]$. 
Therefore, Theorem~\ref{thm:bounded-synth-hard} now follows.

\begin{example}
	We illustrate the above reduction on an example. Suppose we have the formula
	\begin{equation}\label{eq:example2}
		\exists x_1, x_2 \ \forall y_1, y_2 \ (x_1 \land y_1 \land \lnot x_2) \lor (\lnot y_1 \land y_2 \land \lnot x_2) 
		\lor (\lnot y_2 \land x_1)
	\end{equation}
	Let $D_1$, $D_2$ and $D_3$ be the sub-formulas $x_1 \land y_1 \land \lnot x_2$, $\lnot y_1 \land y_2 \land \lnot x_2$ and $\lnot y_2 \land x_1$ respectively. Correspondingly, we get the terms
	$E_1 = (1-s_1) + (1-t_1) + s_2$, $E_2 = t_1 + (1-t_2) + s_2$ and $E_3 = t_2 + (1-s_1)$. 
	Now, for the formula~\ref{eq:example2}, our reduction produces the sketch given in Figure~\ref{fig:hardness}.
	\begin{figure}
		\begin{center}
				\tikzstyle{node}=[circle,draw=black,thick,minimum size=7mm,inner sep=0.75mm,font=\normalsize]
				\tikzstyle{edgelabelabove}=[sloped, above, align= center]
				\tikzstyle{edgelabelbelow}=[sloped, below, align= center]
				\begin{tikzpicture}[->,node distance = 0.5cm,scale=0.8, every node/.style={scale=0.8},initial text = {}]
						\node[node, initial] (start0) {{$start_{0}$}};
						\node[node, right = of start0, xshift = 3.25cm] (start1) {{$start_{1}$}};
						\node[node, right = of start1, xshift = 2.5cm] (start2) {{$start_{2}$}};
						\node[node, right = of start2, xshift = 1.8cm] (start3) {{$start_{3}$}};
						
						\draw(start0) edge[edgelabelabove] node{{\scriptsize $v < (1-s_1) + (1-t_1) + s_2$}} (start1);
%
						\draw(start1) edge[edgelabelabove] node{\scriptsize{$v < t_1 + (1-t_2) + s_2$}} (start2);
%
						\draw(start2) edge[edgelabelabove] node{\scriptsize{$v < t_2 + (1-s_1)$}} (start3);
%
					\end{tikzpicture}
			\end{center}
		\caption{Sketch for formula~\ref{eq:example2}}
		\label{fig:hardness}
	\end{figure}
	
	Note that the formula is true. Indeed, suppose we set $x_1$ to true and $x_2$ to false.
	Then, $D_1$ can be made false only by setting $y_1$ to false, and $D_2$ can be made false
	only by setting either $y_1$ to true or $y_2$ to false, and $D_3$ can be made false
	only by setting $y_2$ to false. It follows then that for any assignment to $y_1$ and $y_2$,
	at least one of $D_1$ or $D_2$ or $D_3$ is true. 
	
	Correspondingly, in the sketch, if we set $s_1$ to 1 and $s_2$ to 0, we cannot cover $start_3$.
	Indeed, with this assignment to the indeterminates, the first guard becomes $v < 1 - t_1$,
	the second guard becomes $v < t_1 + 1 - t_2$ and the third guard becomes $v < t_2$.
	Since $v$ is never incremented, in order to cover $start_3$, the inequalities
	$0 < 1 - t_1, \ 0 < t_1 + 1 - t_2$ and $0 < t_2$ must all be simultaneously satisfied.
	However, if the first inequality is true, then $t_1 = 0$ and so we must have
	$0 < 1 - t_2$ and $0 < t_2$, which cannot be both true. Hence, $start_3$ cannot be covered
	if we set $s_1$ to 1 and $s_2$ to 0.

\end{example}

\subsection{A Special Case of Bounded Coefficient Synthesis}

Motivated by the shape of threshold guards appearing in practice, we now consider a special class of sketch threshold automata for which we can obtain better bounds for the bounded coefficient synthesis problem (against coverability properties). We first describe this special class and then state our results.

A sketch is said to have no indeterminate fall guards if for every fall guard $b \cdot x < a_0 + \sum_{1 \le i \le |\paraset|} a_i \cdot p_i$ of the sketch,
all of the entries in $\{b,a_0,a_1,\dots,a_{|\paraset|}\}$ are integers and not indeterminates. 
As mentioned in~\cite{FSTTCS20}, shared variables in threshold automata are typically used for two things: To record the number of messages
of a specific type that has been broadcasted and to keep track of the number of processes crashed so far. If a shared variable $v$ is
used for the first purpose, then all guards containing $v$ are typically rise guards. If $v$ is used for the latter purpose, then we will usually only have a fall guard of the form
$v < f$ where $f$ is the maximum number of processes allowed to crash. This means that there is no need to synthesize coefficients for fall guards in these cases.
Indeed, for almost all of the benchmarks from~\cite{LKWB17:opodis}, fall guards are of this type,
and hence the subclass that we consider here is interesting from a practical point of view.
We now show that 
\begin{theorem}\label{thm:guard-free}
	The bounded coefficient synthesis problem (against coverability properties) for threshold automata with no indeterminate fall guards 
	is \coNP-complete.
\end{theorem}

Hardness follows from the fact that checking if a location is coverable in a threshold automaton with no fall guards is \NP-hard~\cite{FSTTCS20}. Hence we concentrate on proving the
upper bound.

Let $(\Env,\TA,\ell, I)$ be an input for the bounded coefficient synthesis problem where $\TA$ has no indeterminate fall guards and $I = [A,B]$ with $A,B \in \zn$. Given two assignments $X$ and $Y$ to the indeterminates of $\TA$, we say that $X \le Y$ if $X(s) \le Y(s)$ for every indeterminate $s$.
Let $max$ be the assignment given by $max(s) = B$ for every indeterminate $s$.
We have the following proposition whose proof follows from the definition of our subclass.

\begin{proposition}\label{prop:max-assignment}
	Suppose $X, Y$ are assignments to the indeterminates of $\TA$ such that $X \le Y$.
	If a rule $r$ is enabled at a configuration $\sigma$ in $\TA[Y]$, then $r$ is also enabled at $\sigma$ in $\TA[X]$.	
\end{proposition}

\begin{proof}
	Let $r$ be enabled at a configuration $\sigma$ in $\TA[Y]$ and let $g_1,\dots,g_k$ be the 
	set of threshold guards appearing in $r$. For each $1 \le i \le k$, consider the guard $g_i$. 
	\begin{itemize}
		\item Suppose $g_i$ is a rise guard. Then, since $X \le Y$, it follows that every indeterminate
		in $g_i$ is assigned a lower value in $X$ than in $Y$. Since, $g_i$ is true in $\TA[Y]$ at the
		configuration $\sigma$, it then follows that $g_i$ is also true in $\TA[X]$ at $\sigma$.
		\item Suppose $g_i$ is a fall guard. Then, since fall guards in $\TA$ do not have any indeterminates, it follows that if $g_i$ is true in $\TA[Y]$ at $\sigma$, then it is also
		true in $\TA[X]$ at $\sigma$.
	\end{itemize}
	This completes the proof.
\end{proof}

The above proposition gives us the following useful corollary.

\begin{corollary}
	Suppose $\rho$ is a run of $\TA[Y]$. Then for every $X$ such that $X \le Y$,
	$\rho$ is also a run of $\TA[X]$. Consequently, either $\ell$ is not coverable in $\TA[max]$
	or $\ell$ is coverable for any assignment to the indeterminates within the range $I$.
\end{corollary}

This means that for this class, we can simply reduce bounded coefficient synthesis against coverability properties to checking
coverability: Given a sketch $\TA$, compute $\TA[max]$ and check if it covers $\ell$ or not.
Since in \NP, we can check if $\TA[max]$ can cover $\ell$~\cite[Corollary 1]{Complexity},
this then proves Theorem~\ref{thm:guard-free}.

\section{Conclusion}\label{sec:concl}
We have shown that the coefficient synthesis problem for threshold automata is undecidable, even when the given sketch threshold automaton is acyclic and the
violation is given by a coverability property. 
This already implies that if we have a class of properties
capable of expressing the coverability properies, then the coefficient synthesis problem generalized
to that class is also undecidable. For instance, this implies that coefficient synthesis for the class of properties from the $\ELTLFT$ logic~\cite{ELTLFT}, which has been used
to express various properties of threshold automata obtained from distributed algorithms, is also undecidable.
By the same discussion, our results also imply that bounded coefficient synthesis against properties from $\ELTLFT$
is also $\Sigma_2^p$-hard. (However, this result was already known~\cite[Theorem 8]{Complexity} and our main contribution towards this lower bound in this paper was to prove it for coverability properties).
Finally, since our upper bound result pertains to an efficient encoding of the reachability relation into
existential Presburger arithmetic, it can be used for bounded coefficient synthesis against other classes of properties as well.

As part of future work, it might be interesting to study the (bounded) coefficient synthesis problem
when the assignments to the indeterminates are forced to satisfy a property called \emph{multiplicativity}. 
The usefulness of this property stems from the fact that it has been utilized to get some efficient model-checking algorithms for threshold automata~\cite{ELTLFT,Complexity,FSTTCS20}.
It might also be interesting to study these problems for the case when fall guards do not have indeterminates, which as observed in the previous section can be motivated by cases occurring in the benchmarks from literature.

\paragraph*{Acknowledgments. } This project has received funding from the European Research Council (ERC) under the European Union's Horizon 2020 research and innovation programme under grant agreement No 787367 (PaVeS). This research was sponsored in part by the Deutsche Forschungsgemeinschaft project \href{https://gepris.dfg.de/gepris/projekt/389792660}{389792660} TRR 248–CPEC.

\bibliographystyle{splncs04}
\bibliography{References}

\begin{thebibliography}{10}
\providecommand{\url}[1]{\texttt{#1}}
\providecommand{\urlprefix}{URL }
\providecommand{\doi}[1]{https://doi.org/#1}

\bibitem{AlurHV93}
Alur, R., Henzinger, T.A., Vardi, M.Y.: Parametric real-time reasoning. In:
  Kosaraju, S.R., Johnson, D.S., Aggarwal, A. (eds.) Proceedings of the
  Twenty-Fifth Annual {ACM} Symposium on Theory of Computing, May 16-18, 1993,
  San Diego, CA, {USA}. pp. 592--601. {ACM} (1993). \doi{10.1145/167088.167242}

\bibitem{FSTTCS20}
Balasubramanian, A.R.: Parameterized complexity of safety of threshold
  automata. In: Saxena, N., Simon, S. (eds.) 40th {IARCS} Annual Conference on
  Foundations of Software Technology and Theoretical Computer Science, {FSTTCS}
  2020, December 14-18, 2020, {BITS} Pilani, {K} {K} Birla Goa Campus, Goa,
  India (Virtual Conference). LIPIcs, vol.~182, pp. 37:1--37:15. Schloss
  Dagstuhl - Leibniz-Zentrum f{\"{u}}r Informatik (2020).
  \doi{10.4230/LIPICS.FSTTCS.2020.37},
  \url{https://doi.org/10.4230/LIPIcs.FSTTCS.2020.37}

\bibitem{Complexity}
Balasubramanian, A.R., Esparza, J., Lazi{\'{c}}, M.: Complexity of verification
  and synthesis of threshold automata. In: Hung, D.V., Sokolsky, O. (eds.)
  Automated Technology for Verification and Analysis - 18th International
  Symposium, {ATVA} 2020, Hanoi, Vietnam, October 19-23, 2020, Proceedings.
  Lecture Notes in Computer Science, vol. 12302, pp. 144--160. Springer (2020).
  \doi{10.1007/978-3-030-59152-6\_8},
  \url{https://doi.org/10.1007/978-3-030-59152-6\_8}

\bibitem{BKLW19}
Bertrand, N., Konnov, I., Lazi{\'{c}}, M., Widder, J.: Verification of
  randomized consensus algorithms under round-rigid adversaries. Int. J. Softw.
  Tools Technol. Transf.  \textbf{23}(5),  797--821 (2021).
  \doi{10.1007/S10009-020-00603-X},
  \url{https://doi.org/10.1007/s10009-020-00603-x}

\bibitem{AdHoc}
Delzanno, G., Sangnier, A., Zavattaro, G.: Parameterized verification of ad hoc
  networks. In: Gastin, P., Laroussinie, F. (eds.) {CONCUR} 2010 - Concurrency
  Theory, 21th International Conference, {CONCUR} 2010, Paris, France, August
  31-September 3, 2010. Proceedings. Lecture Notes in Computer Science,
  vol.~6269, pp. 313--327. Springer (2010).
  \doi{10.1007/978-3-642-15375-4\_22},
  \url{https://doi.org/10.1007/978-3-642-15375-4\_22}

\bibitem{GermanS92}
German, S.M., Sistla, A.P.: Reasoning about systems with many processes. J.
  {ACM}  \textbf{39}(3),  675--735 (1992). \doi{10.1145/146637.146681},
  \url{https://doi.org/10.1145/146637.146681}

\bibitem{KW18}
Konnov, I., Widder, J.: {ByMC}: {B}yzantine {M}odel {C}hecker. In: Margaria,
  T., Steffen, B. (eds.) Leveraging Applications of Formal Methods,
  Verification and Validation. Distributed Systems - 8th International
  Symposium, ISoLA 2018, Limassol, Cyprus, November 5-9, 2018, Proceedings,
  Part {III}. Lecture Notes in Computer Science, vol. 11246, pp. 327--342.
  Springer (2018). \doi{10.1007/978-3-030-03424-5\_22},
  \url{https://doi.org/10.1007/978-3-030-03424-5\_22}

\bibitem{ELTLFT}
Konnov, I.V., Lazi{\'{c}}, M., Veith, H., Widder, J.: A short counterexample
  property for safety and liveness verification of fault-tolerant distributed
  algorithms. In: Castagna, G., Gordon, A.D. (eds.) Proceedings of the 44th
  {ACM} {SIGPLAN} Symposium on Principles of Programming Languages, {POPL}
  2017, Paris, France, January 18-20, 2017. pp. 719--734. {ACM} (2017).
  \doi{10.1145/3009837.3009860}, \url{https://doi.org/10.1145/3009837.3009860}

\bibitem{KVW17:IandC}
Konnov, I.V., Veith, H., Widder, J.: On the completeness of bounded model
  checking for threshold-based distributed algorithms: Reachability. Inf.
  Comput.  \textbf{252},  95--109 (2017). \doi{10.1016/J.IC.2016.03.006},
  \url{https://doi.org/10.1016/j.ic.2016.03.006}

\bibitem{AllFlavors}
Kukovec, J., Konnov, I., Widder, J.: Reachability in parameterized systems: All
  flavors of threshold automata. In: Schewe, S., Zhang, L. (eds.) 29th
  International Conference on Concurrency Theory, {CONCUR} 2018, September 4-7,
  2018, Beijing, China. LIPIcs, vol.~118, pp. 19:1--19:17. Schloss Dagstuhl -
  Leibniz-Zentrum f{\"{u}}r Informatik (2018).
  \doi{10.4230/LIPICS.CONCUR.2018.19},
  \url{https://doi.org/10.4230/LIPIcs.CONCUR.2018.19}

\bibitem{LKWB17:opodis}
Lazi{\'{c}}, M., Konnov, I., Widder, J., Bloem, R.: Synthesis of distributed
  algorithms with parameterized threshold guards. In: Aspnes, J., Bessani, A.,
  Felber, P., Leit{\~{a}}o, J. (eds.) 21st International Conference on
  Principles of Distributed Systems, {OPODIS} 2017, Lisbon, Portugal, December
  18-20, 2017. LIPIcs, vol.~95, pp. 32:1--32:20. Schloss Dagstuhl -
  Leibniz-Zentrum f{\"{u}}r Informatik (2017).
  \doi{10.4230/LIPICS.OPODIS.2017.32},
  \url{https://doi.org/10.4230/LIPIcs.OPODIS.2017.32}

\bibitem{SynthesisOneCounter}
P{\'{e}}rez, G.A., Raha, R.: Revisiting parameter synthesis for one-counter
  automata. In: Manea, F., Simpson, A. (eds.) 30th {EACSL} Annual Conference on
  Computer Science Logic, {CSL} 2022, February 14-19, 2022, G{\"{o}}ttingen,
  Germany (Virtual Conference). LIPIcs, vol.~216, pp. 33:1--33:18. Schloss
  Dagstuhl - Leibniz-Zentrum f{\"{u}}r Informatik (2022).
  \doi{10.4230/LIPICS.CSL.2022.33},
  \url{https://doi.org/10.4230/LIPIcs.CSL.2022.33}

\bibitem{ST87:abc}
Srikanth, T.K., Toueg, S.: Simulating authenticated broadcasts to derive simple
  fault-tolerant algorithms. Distributed Comput.  \textbf{2}(2),  80--94
  (1987). \doi{10.1007/BF01667080}, \url{https://doi.org/10.1007/BF01667080}

\bibitem{CADE2005}
Verma, K.N., Seidl, H., Schwentick, T.: On the complexity of equational horn
  clauses. In: Nieuwenhuis, R. (ed.) Automated Deduction - CADE-20, 20th
  International Conference on Automated Deduction, Tallinn, Estonia, July
  22-27, 2005, Proceedings. Lecture Notes in Computer Science, vol.~3632, pp.
  337--352. Springer (2005). \doi{10.1007/11532231\_25},
  \url{https://doi.org/10.1007/11532231\_25}

\bibitem{Wrathall76}
Wrathall, C.: Complete sets and the polynomial-time hierarchy. Theor. Comput.
  Sci.  \textbf{3}(1),  23--33 (1976). \doi{10.1016/0304-3975(76)90062-1},
  \url{https://doi.org/10.1016/0304-3975(76)90062-1}

\end{thebibliography}

%
%
%

%
\end{document}